\documentclass[lettersize,journal]{IEEEtran} 

\usepackage{mathtools}
\usepackage{float}
\usepackage{multirow}
\usepackage{makecell} 
\usepackage[font=small,labelfont=bf,labelsep=space]{caption}
\usepackage{graphicx} 
\graphicspath{ {images/} }

\usepackage[linesnumbered,ruled,vlined]{algorithm2e} 

\usepackage[utf8]{inputenc}
\usepackage{amsthm,amssymb}

\usepackage{optidef}

\theoremstyle{definition}
\newtheorem{example}{Example}

\newtheorem{proposition}{Proposition}
\newtheorem{lemma}{Lemma}

\newtheorem{definition}{Definition}
\newtheorem{assumption}{Assumption}

\newtheorem{remark}{Remark}

\DeclareMathOperator{\sgn}{sgn}
\DeclareMathOperator{\supp}{supp}

\newcommand{\cC}{\mathcal{C}}
\newcommand{\cE}{\mathcal{E}}
\newcommand{\bG}{\mathbf{G}}
\newcommand{\bH}{\mathbf{H}}
\newcommand{\be}{\mathbf{e}}
\newcommand{\bz}{\mathbf{z}}
\newcommand{\bhe}{\hat{\mathbf{e}}}

\newcommand{\br}{\mathbf{r}}
\newcommand{\bp}{\mathbf{p}}
\newcommand{\bh}{\mathbf{h}}
\newcommand{\bs}{\mathbf{s}}
\newcommand{\bhc}{\hat{\mathbf{c}}}
\newcommand{\bc}{\mathbf{c}}

\newcommand{\bv}{\mathbf{v}}

\newcommand{\cS}{\mathcal{S}}
\newcommand{\cP}{\mathcal{P}}

\newcommand{\bzero}{\mathbf{0}}
\newcommand{\bone}{\mathbf{1}}
\newcommand{\ft}{\mathbb{F}_2}
\newcommand{\cI}{\mathcal{I}}

\usepackage{xcolor}
\def\BibTeX{{\rm B\kern-.05em{\sc i\kern-.025em b}\kern-.08em
    T\kern-.1667em\lower.7ex\hbox{E}\kern-.125emX}}

\usepackage{fancyhdr}

\IEEEoverridecommandlockouts                              

\title{
Segmented GRAND: Complexity Reduction through Sub-Pattern Combination 
}
\author{
\IEEEauthorblockN{Mohammad Rowshan 
and Jinhong Yuan, {\em Fellow, IEEE}}
 \thanks{Mohammad Rowshan and Jinhong Yuan are with the school of electrical engineering and Telecommunications, University of New South Wales (UNSW) in Sydney, Australia. (e-mail: \{m.rowshan,j.yuan\}@unsw.edu.au).}
  \thanks{The work was supported in part by the Australian Research Council (ARC) Discovery Project under Grant DP220103596, and in part by the ARC Linkage Project under Grant LP200301482.}
 \thanks{This paper was presented in part at the 2023 IEEE Global Communication Conference (GLOBECOM), Kuala Lumpur, Malaysia \cite{rowshan2023low}.}
}

\begin{document}

\maketitle
\thispagestyle{fancy}
\fancyhf{}
\lhead{Accepted for publication in a forthcoming issue of IEEE Transactions on Communications. This version contains two additional sections.}
\cfoot{}

\begin{abstract} 
The ordered-reliability bits (ORB) variant of guessing random additive noise decoding (GRAND), known as ORBGRAND, achieves remarkably low time complexity at high code rates compared to other GRAND variants. However, its computational complexity remains higher than other near-ML universal decoders like ordered-statistics decoding (OSD). To address this, we propose segmented ORBGRAND, which partitions the error pattern search space based on code properties, generates syndrome-consistent sub-patterns (reducing invalid error patterns), and combines them in a near-ML order using sub-weights derived from two-level integer partitions of logistic weight. Numerical results show that segmented ORBGRAND reduces the average number of queries by at least 66\% across all SNRs and cuts basic operations by over an order of magnitude, depending on segmentation and code rate. Further efficiency gains come from leveraging pre-generated shared sub-patterns, reducing average decoding time. Furthermore, with abandonment ($b=10^{5}$ or smaller), segmented ORBGRAND provides a 0.2 dB power gain over ORBGRAND. Additionally, we provide an analytical justification for why the logistic weight-based ordering of error patterns in ORBGRAND closely approximates the ML order and discuss the underlying assumptions of ORBGRAND.  

\end{abstract}



\section{INTRODUCTION}
\label{sec:intro}
Soft decision-based decoding algorithms can be classified into two major categories \cite{lin}: Code structure-based algorithms and reliability-based algorithms or generic decoding algorithms as they usually do not depend on the code structure. In the generic algorithms a.k.a universal algorithms, which is the focus of this paper, the goal is to find the closest modulated codeword to the received sequence using a metric such as the likelihood function. That is, we try to maximize the likelihood in the search towards finding the transmitted sequence. Hence, this category of decoding algorithm is called maximum likelihood (ML) decoding which is known as an optimal decoding approach. Maximum likelihood decoding has been an attractive subject for decades among researchers. Error sequence generation is one of the central problems in any ML decoding scheme. 
The brute-force approach for ML decoding of a linear $(n,k)$ block code requires the computation of likelihood or Euclidean distances of $2^k$  modulated codewords from the received sequence. In general, ML decoding is prohibitively complex for most codes as it was shown to be an NP-complete problem \cite{berlekamp}. 
Hence, the main effort of the researchers has been concentrated on reducing the algorithm's complexity for short block-lengths. Although there are approaches in which optimal performance is preserved, ML performance can be traded off for significant complexity reduction. Here, we review some of the notable efforts toward complexity reduction in the past decades. 

Forney proposed the generalized minimum distance (GMD) decoding algorithm in 1966 \cite{forney}, where a list of candidate codewords based on the reliability of the received symbols was produced using an algebraic decoder. 
In 1972, Chase proposed a method \cite{chase} in which the search was performed among a fixed number of error patterns corresponding to a particular number of least reliable bit positions with respect to the minimum distance $d$ of the underlying code. 
Chase classified his algorithm into three types, as per the error pattern generation. 
In another effort, Snyders in 1989 \cite{snyders} proposed to perform syndrome decoding on the received sequence and then use the syndrome information to modify and improve the original hard-decision decoding. 

The best-known generic decoding algorithm is perhaps the information set decoding (ISD) algorithm proposed by Prange in 1962 \cite{prange}, which was improved 
by Stern in 1989 \cite{stern} and Dumer in 1991 \cite{dumer}. 
Following this approach, other generic decoding approaches were developed based on the most reliable basis (MRB), defined as the support of the most reliable independent positions (MRIPs) of the received sequence, hence forming an information set. In these approaches, each error pattern is subtracted from the hard decision of the MRIPs and the corresponding codeword is reconstructed by encoding the corresponding information sequence. In 1974, Dorsch \cite{dorsch} considered error patterns restricted to the MRB in increasing a priori likelihood. 
Following this approach, Fossorier and Lin in 1995 \cite{fossorier} proposed processing the error patterns in a deterministic order within families of increasing Hamming weight. This algorithm, which is referred to as ordered statistics decoding (OSD), is  
one of the most popular generic decoding algorithms nowadays. The OSD algorithm permutes the columns of the generator matrix with respect to the reliability of the symbols for every received vector and performs elementary row operations on the independent columns extracted from the permuted generator matrix resulting in the systematic form.  
The testing error patterns can have a Hamming weight of up to $l, 0\leq i \leq k$ in $i$-order OSD, chosen from the most reliable $k$ positions. 
Apparently, the main drawback of OSD is the use of row operations to put either the generator matrix or the parity check matrix of the code into systematic form. The complexity of row operation for an $(n,k)$ linear block code is $O(n^3 \min\{R, 1-R\}^2)$ where R is the code rate.  
However, since overall complexity is an exponential function of code length, this preprocessing complexity is negligible.  Moreover, having information set in a systematic form is needed only for simplifying further the decoding attempts. Otherwise, the error patterns can be checked without this preprocessing. 
The OSD algorithm further evolved in 2004 into box-and-match algorithm (BMA) \cite{valembois} and enhanced BMA \cite{jin} where the matching technique was used to reduce time complexity at the cost of space complexity. There are efficient and fast hardware implementation available for OSD such as \cite{kim2021fpga} which reduces the latency by 12 times. The matching techniques were employed for fast decoding of polar codes with Reed-Solomon kernel in \cite{trifonov18match}. 
It is worth noting that a similar algorithm  to BMA, called the sort and match algorithm, was proposed by Dumer in 1991 in \cite{dumer_SMA,dumer_SMA_punc} which has the same asymptotic complexity as BMA. 

In 2018, Duffy et al. \cite{duffy18guess} suggested a hard-decision scheme in which the error patterns were ordered from most likely to least likely ones based on a statistical channel model, and then the incorrect error patterns were sequentially removed, querying for the first error patterns corresponding to a valid codeword. This original idea, which was later called \emph{guessing random additive noise decoding} (GRAND), further developed into a soft decision scheme or sGRAND where the error patterns were generated based on the symbols' reliability  and sequential insertion and removal of the error patterns from an ordered stack until the first valid codeword was found. The sGRAND was shown to be capacity achieving \cite{duffy-tit} and ML algorithm \cite{duffy-sgrand} though it came at a significant computational complexity cost because error patterns in the stack needed to be sorted after insertion of new patterns into the stack. The approach used in GRAND appears to be in line with a general optimum technique proposed in \cite{valembois01weightFunc} to handle pattern generation while maintaining monotonicity \cite{fossorier_grand}. Recently, a reduced complexity variant of sGRAND was proposed in \cite{zheng-cgd} that limits the scope of guessing by breaking down the error coordinates into the parity and information coordinates. 
The next evolution in this approach occurred by employing a simple metric that gave error patterns for testing in a near ML order \cite{duffy-orbgrand}. This step was a significant boost for GRAND toward making it practical for high rate and short codes.  The approximate scheduling of the error sequences is based on a distinct integer partitioning of positive integers, which is significantly less complex. Alternatively, a sequential algorithmic method to generate error sequences was suggested based on partial ordering and a modified logistic weight in \cite{condo-grand} that prioritizes the low-weight error sequences resulting in improving the performance though its pattern generation process is not as simple as the integer partitioning process. In \cite{liu}, it was shown that ORBGRAND is almost capacity-achieving and a slight improvement in the block error rate, in particular at high SNR regimes, was demonstrated based on an information-theoretic study. It is worth noting that there also exists a variant of GRAND that provides block-wise soft output to control the decoding misdetection rate and bitwise soft output for efficient iterative decoding \cite{duffy2024soft,yuan-sogrand}. 
Several hardware architectures have also been proposed for ORBGRAND in \cite{abbas-orbgrand,abbas-listgrand,riaz,condo-lutgrand}.


The main advantage of ORBGRAND is its simplicity in generating error patterns in an order near ML by a simple weight function that makes it a hardware-friendly algorithm. Unlike some of the other schemes, it does not require any preprocessing, or sorting (except for the reliability order) and it has inherently an early termination mechanism in itself that stops searching after finding the most likely codeword or near that.  However, the number of invalid error patterns is significantly high. The aim of this work and our previous work in \cite{rowshan-const_GRAND} was to reduce invalid patterns and save computations and time. In constrained GRAND \cite{rowshan-const_GRAND}, by simply utilizing the structure of a binary linear code, 
we proposed an efficient pre-evaluation that constrains the error pattern generation. This approach could save the codebook checking operation. 
These syndrome-based constraints are extracted from the parity check matrix (with or without matrix manipulation) of the underlying code. 
We also showed that the size of the search space is deterministically reduced by a factor of $2^p$ where $p$ is the number of constraints. 
Note that constrained error sequence generation does not degrade the error correction performance as it is just discarding the error sequences that do not result in valid codewords. The proposed approach could be applied to other GRAND variants such as SGRAND \cite{duffy-sgrand}. 

In this paper, different from \cite{rowshan-const_GRAND}, we propose an approach that generates sub-patterns for the segments corresponding to the defined constraints. 
We simultaneously generate sub-patterns for each segment with odd or even weight, guided by the available information from the syndrome, otherwise with both weights. 
To address the challenging problem of combining the sub-patterns in an ML-order, we propose a customized partition (a.k.a composition  \cite{heubach}) of the logistic weight into segment-specific sub-weights. This composition involves partitioning the logistic weight into non-distinct positive integers, with the number of parts (a.k.a composition order) restricted to the number of segments. Furthermore, our approach allows zero to be included as an element in the composition. 
The numerical results show that by employing the proposed method, the average number of attempts is significantly reduced compared to ORBGRAND by at least 66\% depending on the number of segments. This reduction is justified by the reduction in the size of the search space. Furthermore, this approach can improve the block error rate when employing segmented ORBGRAND with abandonment (by more than 0.2 dB power gain when the abandonment threshold is $b=10^{5}$) as segmented ORBGRAND effectively increases the abandonment threshold $b$. However, this gain diminishes as the abandonment threshold increases.

\section{PRELIMINARIES}\label{sec:prelim}
We denote by $\ft$ the binary finite field with two elements. The cardinality of a set is denoted by $|\cdot|$.  The interval $[a,b]$ represents the set of all integer numbers in $\{x:a\leq x\leq b\}$. The \emph{support} of a vector $\be = (e_1,\ldots,e_{n}) \in \ft^n$ is the set of indices where $\be$ has a nonzero coordinate, i.e. $\supp(\be) \triangleq \{i \in [1,n] \colon e_i \neq 0\}$ . 
The \emph{weight} of a vector $\be \in \ft^n$ is $w(\be)\triangleq |\supp(\be)|$. The all-one vector $\bone$ and all-zero vector~ $\bzero$ are defined as vectors with all identical elements of 1 or 0, respectively. The summation in $\ft$ is denoted by $\oplus$. The modulo operation (to obtain the remainder of a division) is denoted by $\%$. 
\subsection{ML Decoding and Ordered Reliability Bits GRAND}
A binary code $\mathcal{C}$ of length $n$ and dimension $k$ maps a message of $k$ bits to a codeword $\mathbf{c}$ of $n$ bits to be transmitted over a noisy channel. We assume that we are using binary phase shift keying (BPSK) modulation. The channel alters the transmitted codeword so that the receiver obtains an $n$-symbol vector~$\mathbf{r}$. A ML decoder supposedly compares $\mathbf{r}$ with all the $2^k$ modulated codewords in the codebook, and selects the one closest to $\mathbf{r}$. In other words, the ML decoder finds a modulated codeword $\textup{x}(\bc)$ such that
\begin{equation}\label{eq:likelihood}
    \hat{\mathbf{c}} = \underset{\mathbf{c}\in\mathcal{C}}{\text{arg max }} p\big(\mathbf{r}|\textup{x}(\mathbf{c})\big).
\end{equation}

For additive white Gaussian noise (AWGN) channel with noise power of $\sigma^2_n=N_0/2$ where $N_{0}$ is the noise spectral density, the conditional probability  $p\big(\mathbf{r}|\textup{x}(\mathbf{c})\big)$ is given by 
\begin{equation}\label{eq:cond_prob}
    p\big(\mathbf{r}|\textup{x}(\mathbf{c})\big)= \frac{1}{(\sqrt{\pi N_0})^n}\text{exp}\Bigg(-\sum_{i=1}^n(r_i-\textup{x}(c_i))^2/N_0\Bigg).
\end{equation}
Observe that maximizing $p(\mathbf{r}|\textup{x}(\mathbf{c}))$ is equivalent to minimizing 
\begin{equation}\label{eq:sq_euclid_dist}
d^2_E=\sum_{i=1}^n(r_i-\textup{x}(c_i))^2,
\end{equation}
which is called {\em squared Euclidean distance} (SED). Therefore, we have
\begin{equation}\label{eq:ML_dec}
    \hat{\mathbf{c}} = \underset{\mathbf{c}\in\mathcal{C}}{\text{arg max }} p\big(\mathbf{r}|\textup{x}(\mathbf{c})\big)=\underset{\mathbf{c}\in\mathcal{C}}{\text{arg min }}\big(\mathbf{r}-\textup{x}(\mathbf{c})\big)^2.
\end{equation}

The process of finding $\mathbf{c}$, depending on the scheme we employ, may require checking possibly a large number of binary error sequences $\bhe$ to select the one that satisfies 
\begin{equation}\label{eq:check}
     \mathbf{H} \cdot (\theta(\mathbf{r})\oplus\hat{\mathbf{e}}) = \mathbf{0}
\end{equation}
where $\theta(\mathbf{r})$ returns the hard-decision demodulation of the received vector $\mathbf{r}$ and $\bH$ is the parity check matrix of code $\cC$, 
$
    \bH = [
        \bh_1\,
        \bh_2\,
        \cdots\,
        \bh_{n-k}
    ]^T
$
and the $n$-element row vectors $\bh_j$ for $j\in[1,n-k]$ are denoted by $\bh_j = [h_{j,1} \; h_{j,2}  \;\cdots \;h_{j,n}]$. Note that any valid codeword $\bc=\theta(\mathbf{r})\oplus\hat{\mathbf{e}}$  gives $\bH\cdot\bc=\bzero$. Here, $\bhe$ is the binary error sequence to which we refer as an error pattern. 

To obtain the error patterns in ML order, one can 1) generate all possible error patterns $\bhe$, that is, $\sum_{j=1}^{n}{n \choose j}$ patterns, 2) sort them based on a likelihood measure such as the squared Euclidean distance $\big(\mathbf{r}-\textup{x}\big(\theta(\mathbf{r})+\bhe\big)\big)^2$, and then 3) check them using \eqref{eq:check} one by one from the smallest distance in ascending order. It was numerically shown in \cite{duffy-orbgrand} that 
the error patterns generated by 
all the integer partitions of {\em logistic weights} $w_L=1,2,\dots, n(n+1)/2$ can give an order close to what we describe earlier. Obviously, the latter method, which is used in ORBGRAND, is more attractive, as it does not need any sorting operation on a large set of metrics at every decoding step. Note that in conventional ML decoding, the test error patterns are not checked in ML order, and thus no sorting is required. The scenario described above was a hypothetical and impractical scenario that has been made possible using logistic weight.

The logistic weight $w_L$ of a length-$n$ binary vector $\bz$ is defined as \cite{duffy-orbgrand} 
\begin{equation}\label{eq:w_L}
    w_L(\bz)=\sum_{i=1}^n z_i\cdot i
\end{equation}
where $z_i\in\ft$ is the $i$-th element of the error pattern $\hat{\mathbf{e}}$ permuted in the ascending order of the received symbols' reliability 
$|r_i|,i\in[1,n]$. That is, the error pattern is $\bhe=\pi(\mathbf{z})$ where $\pi(\cdot)$ is the vector-wise permutation function that maps the binary vector $\bz$ to the error pattern $\bhe$. For the element-wise mapping of this permutation,  we will use $\dot{\pi}(\cdot)$ for mapping the index of any element in $\bz$ to the corresponding element in $\hat{\mathbf{e}}$ and $\dot{\pi}^{-1}(\cdot)$ for the reverse mapping. For the sake of simplicity, we refer to $w_L(\bz)$ by $w_L$.

To obtain all binary vectors $\bz$ corresponding to a certain $w_L$, there is a simple approach. All coordinates $j$ in $\bz$, where $z_j=1$ for a certain $w_L$, can be obtained from {\em integer partitions} of $w_L$ with distinct parts and no parts larger than the code length$n$. 
Let us define the integer partitions of $w_L$ mathematically as follows:
\begin{definition}\label{def:int_part}
The integer partitions of $w_L$ form the set of subsets $\cI\subset[1,w_L]$ such that 
\begin{equation}
    w_L=\sum_{j\in\cI\subset[1,w_L]}j. 
\end{equation}
Then, the binary vector $\bz$ corresponding to any $\cI$ consists of the elements $z_j=1,j\in\cI$ and $z_j=0,j\not\in\cI$. 
\end{definition}
In Definition \ref{def:int_part}, we abused the notion of integer partitions and considered a single part/partition as well to cover all the error patterns obtained from every $w_L$. Observe that for every $w_L$, there exists at least one $\cI$ with a single element $w_L$.  For instance, for $w_L=1,2$, we have a single $\cI=\{w_L\}$. As $w_L$ gets larger, the number of subsets $\cI\subset[1,w_L]$ increases. 

\begin{example}\label{ex:permut}
    Suppose we have the received sequence $\br=[0.5,-1.2,0.8,1.8,-1,-0.2,0.7,-0.9]$. 
    
    We can get the following permutation based on $|r_i|,i\in[1,8]$ in ascending order:
    \[
        \dot{\pi} : [1,2,3,4,5,6,7,8]\rightarrow[6,1,7,3,8,5,2,4]
    \]
    Assuming we have attempted all the error patterns generated based on $w_L=1,2,3,4,5$ so far. Then, we need to find the error patterns based on $w_L=6$. The integer partitions of $w_L=6$ are $\cI=\{6\},\{1,5\},\{2,4\}$, and $\{1,2,3\}$ that satisfy $w_L=\sum_{j\in\cI}j,\cI\subset[1,6]$. We call every element in $\cI$ as a \emph{part}. Then, the $\bz$ vector and the corresponding error patterns after the vector permutation $\pi$ are
    $$\bz=[0\;0\;0\,0\;0\;1\;0\;0]\rightarrow \bhe=[0\;0\;0\;0\;1\;0\;0\;0],$$
    $$\bz=[1\;0\;0\;0\;1\;0\;0\;0]\rightarrow \bhe=[0\;0\;0\;0\;0\;1\;0\;1],$$
    $$\bz=[0\;1\;0\;1\;0\;0\;0\;0]\rightarrow \bhe=[1\;0\;1\;0\;0\;0\;0\;0],$$
    $$\bz=[1\;1\;1\;0\;0\;0\;0\;0]\rightarrow \bhe=[1\;0\;0\;0\;0\;1\;1\;0].$$
    These error patterns can be checked using \eqref{eq:check} in an arbitrary order. In the next section, we will see that any of these error patterns results in an identical increase in $d_E^2$ in \eqref{eq:sq_euclid_dist}, i.e., they are all located at an identical distance from the received sequence, under some assumption about the distribution of $|r_i|,i\in[1,n]$.
\end{example}

\begin{remark}
    By statistically analyzing the reliability of the received sequence or any other insight, one can prioritize small Hamming weights over those with large Hamming weights, or vice versa. Alternatively, we can limit the scope of the attempts to small or large Hamming weights.  Observe that as the logistic weight increases, error patterns with larger Hamming weights will be generated.
\end{remark}

\subsection{Ordered Statistics Decoding (OSD)}\label{ssec:OSD}
Ordered statistics decoding with order $i$ provides an efficient method to find the best possible codeword $\bhc$ in \eqref{eq:ML_dec} after querying among $\sum_{l=0}^i{k\choose l}$ candidate codewords. In OSD, the columns of the generator matrix $\bG$ are first sorted in descending order of reliability based on the received symbols $\br$. This results in the permutation function $\lambda_1(\cdot)$, which rearranges the generator matrix to yield $\bG'=\lambda_1(\bG)$. Next, starting from the first column of $\bG'$, finding $k$ linearly independent columns gives the second permutation function $\lambda_2(\cdot)$, leading to the transformation $\bG''=\lambda_2(\lambda_1(\bG))$ and $\mathbf{y}=\lambda_2(\lambda_1(\br))$. We then convert the generator matrix $\bG''$ to systematic form, $\bG''_{sys}$. 
Now, we proceed with the OSD process by generating candidate codewords as $\bv=(\theta(y_1^k)\oplus z_1^k)\bG''_{sys}$ where $z_1^k$ is a double permuted error vector with $\operatorname{w}(z_1^k)=l$.

Furthermore, we apply the early termination criterion based on the sufficient condition for optimality \cite{lin} where we stop the reprocessing of the order-$i$ OSD at order $l\leq i$ and declare the closest/best found codeword $\mathbf{v}_{\text{best}}$ given that correlation discrepancy metric $\lambda\left(\mathbf{y}, \mathbf{v}_{\text{best}}\right)=\sum_{i \in D_1(\mathbf{v}_{\text{best}})}\left|y_i\right|$ is smaller than the bound on optimality $G$ at the end of reprocessing order-$l$, that is, $\lambda\left(\mathbf{y}, \mathbf{v}_{\text{best}}\right)\leq G\left(\mathbf{v}_{\text {best}}, d_{\min}\right)$
where 
\begin{equation}
    G\left(\mathbf{v}_{\text {best}}, d_{\min}\right)=\sum_{j=0}^{l}\left|y_{k-j}\right| + \sum_{j \in D_0^{(\delta^{\prime})}\left(\mathbf{v}_{\text {best }}\right)}\left|y_j\right|,
\end{equation}
\begin{equation}
\delta^{\prime}=\max \left\{0, d_{\text {min }}-\left|D_1\left(\mathbf{v}_{\text {best}}\right)\right|-(l+1)\right\},
\end{equation}
\begin{equation}
D_1(\mathbf{v}) \triangleq\left\{i: v_i \neq z_i, 1 \leq i\leq n\right\},
\end{equation}
\begin{equation}
D_0(\mathbf{v}) \triangleq [1,n] \backslash D_1(\mathbf{v}),
\end{equation}
and $D_0^{(j)}\left(\mathbf{v}\right)$ gives the set of first $j$ indices in the set $D_0(\mathbf{v})$ that is sorted based on reliability in ascending order.

\section{Near-ML Ordering of Error Patterns with Logistic Weight}
In this section, 
we investigate analytically how the error patterns in the ascending order of the logistic weight can closely follow the maximum likelihood order over the AWGN channel. The analysis is based on an assumption made for ORBGRAND \cite{duffy-orbgrand2}, which is in disagreement with the Gaussian distribution in the AWGN channel. This assumption is also a basis for devising a similar approach for combining the sub-patterns in the segmented ORBGRAND in Section \ref{sec:comb_patterns}. 

\begin{assumption}\label{assump:equidistant}
We assume that the ordered sequence of $|r_i|,i=1,2,\dots,n$ as 
$$|r_1|\leq|r_2|\leq|r_3|\leq\cdots$$ 
are placed equidistantly. That is, $$\delta=|r_{i+1}|-|r_i|=|r_{i+2}|-|r_{i+1}|=\cdots{\color{blue}, \text{ and } \delta>0}.$$ 
Additionally, for some $\rho\geq0$, we define 
$$|r_i|=\rho+i\cdot\delta.$$
\end{assumption}
Now, let us get back to the Euclidean distance. The squared Euclidean distance (SED) as a function of $\bz$ denoted by $d^2_E(\bz)$ is 
\begin{equation}\label{eq:euclid_dist_e}
d^2_E(\bz)=\sum_{i=1}^{n}(r_i-\textup{x}(\theta(r_i)\oplus z_i))^2. 
\end{equation}
and for $\bz=\bzero$, we have 
$$d^2_E(\bzero)=\sum_{i=1}^{n}(r_i-\textup{x}(\theta(r_i)))^2$$ 
which is the minimum SED that we can get. Hence,  
$$d^2_E(\bz)>d^2_E(\bzero)$$ 
and the increase of $d^2_E(\bz)$ compared to $d^2_E(\bzero)$, denoted by $d^{(+)}$, is formulated as
\begin{equation}\label{eq:dE0_add}
d^2_E(\bz) = d^2_E(\bzero) + d^{(+)}(\bz)
\end{equation}
for any $\bz\neq \bzero$. For the sake of simplicity, we refer to $d^{(+)}(\bz)$ by $d^{(+)}$. 

Observe that $\textup{x}(\theta(r_i))=\sgn(r_i)$ and when we apply $z_i=1$, the sign changes as follows
\begin{equation}\label{eq:sgn_change}
    \textup{x}(\theta(r_i)\oplus z_i) =
    \begin{dcases}
        \sgn(r_i) & z_i=0,\\
        -\sgn(r_i) & z_i=1.\\
    \end{dcases}
\end{equation}

Without loss of generality, we assume $r_i>0$ hence $r_i=i\cdot\delta$ for $\rho=0$ any $i$ with $z_i=1$ to make the following discussion easier to follow. Since $\sgn(r_i)\in\{1,-1\}$ is a bipolar mapping, then, we have
\begin{equation}
    (r_i-\textup{x}(\theta(r_i)\oplus z_i))^2=(i\cdot\delta - 1)^2
\end{equation}
To begin with, we consider only the error patterns with a single error. 
For a pattern $\bz$ with $w(\bz)=1$, we flip $z_i=0$ to $z_i=1$ and we get $(i\cdot\delta + 1)^2$. Then, the increase in the SED is
\begin{equation}\label{eq:d_plus_e1}
d^{(+)} = (i\cdot\delta + 1)^2-(i\cdot\delta - 1)^2 = i(4\delta) = i\Delta.
\end{equation}
where the notation $\Delta=4\delta$ is introduced and it will be used in the rest of this section. 

Now, let us take all the error patterns with identical logistic weights. As we know, these patterns can be obtained by integer partitioning with distinct parts. The following proposition discusses the increase in the SED for this case, where the logistic $w_L$ is found proportional to the increase in distance from the received sequence, $d^{(+)}=d^2_E(\bz)-d^2_E(\bzero)$. In other words, given Assumption \ref{assump:equidistant}, our aim is to show that
\begin{equation}\label{eq:propto_wL}
d^{(+)} \propto w_L.
\end{equation}

\begin{proposition}
Given an arbitrary logistic weight $w_L>0$ and Assumption~\ref{assump:equidistant}, the increase in the squared Euclidean distance, i.e., the term $d^{(+)}(\bz)$ in $d^2_E(\bz) = d^2_E(\bzero) + d^{(+)}(\bz)$, remains constant for all binary vector $\bz$  with $z_j=1,j\in\cI\subset[1,w_L]$ such that $w_L=\sum_{j\in\cI}j$. That is, for some $\Delta>0$, we have
\begin{equation}\label{eq:wL_dist}
d^{(+)} = \big(\sum_{j\in\cI}j\big)\Delta \;\; \text{for all $\cI\subset[1,w_L]$ s.t.}\;\; w_L=\sum_{j\in\cI}j.
\end{equation}

\end{proposition}
\begin{proof}
Suppose $w_L=i=i_1+i_2$. We first compare $d^{(+)}$ for the error patterns corresponding to $i$ alone and $i_1,i_2$ together. We observed the increase in the SED by an error pattern with $w(\bz)=1$ in \eqref{eq:d_plus_e1}. 
Now, if we use an error pattern $\bz$ with weight $w(\bz)=2$ by flipping $z_{i_1}=z_{i_2}=0$ to $z_{i_1}=z_{i_2}=1$ given $w_L(\bz)=i=i_1+i_2$, we get 
\begin{multline*}
d^{(+)} = \Big((i_1\delta + 1)^2+(i_2\delta + 1)^2\Big)-\Big((i_1\delta - 1)^2+(i_2\delta - 1)^2\Big)\\ =
\Big((i_1\delta + 1)^2-(i_1\delta - 1)^2\Big)+\Big((i_2\delta + 1)^2-(i_2\delta - 1)^2\Big)\\ \stackrel{\text{\eqref{eq:d_plus_e1}}}{=}
i_1\Delta+i_2\Delta=(i_1+i_2)\Delta
\end{multline*}

In general, if we use any error pattern $\bz$ with weight larger than $w(\bz)>1$ given $w_L(\bz)=i$, we have
\begin{equation}
d^{(+)} = \sum_{j\in\cI}\Big((j\delta + 1)^2-(j\delta - 1)^2\Big) = \big(\sum_{j\in\cI}j\big)\Delta.
\end{equation}
Therefore, as $i=\sum_{j\in\cI}j$, any error pattern $\bz$ with $z_j=1,j\in\cI$ and $\cI\subset[1,i)$ gives the same $d^{(+)}$.
Note that all $\cI$ subsets can be obtained by integer partitioning with distinct parts.
\end{proof}

Hence, test error patterns with an identical logistic weight will have the identical squared Euclidean distance as well. That is why the order of checking these patterns is arbitrary as suggested in \cite{duffy-orbgrand}. 

\begin{remark}
Given two logistic weights of $w_L=i$ and $i'$ such that $i'>i$. Since $i'\Delta>i\Delta$ and so the $d^{(+)}$ corresponding to $i'$ will be larger, we have $d^2_E(\bz)<d^2_E(\bz')$ where $\bz$ and $\bz'$ are the test error patterns corresponding to $w_L=i$ and $i'$. Hence, the test error pattern(s) with $w_L=i$ should be checked first in this case.
\end{remark}

Recall that we considered Assumption \ref{assump:equidistant} for the analysis in this section, which implies a uniform distribution for the received signals. However, this assumption is not realistic as the $r_i$ values follow the Gaussian distribution. 
Therefore, the test error patterns in the order generated based on the logistic weight may not be precisely aligned with the ML order. As a result, we refer to this order as a near-ML order.
\section{Segmented GRAND: Error Sub-patterns} \label{sec:err_subpattern}
In \cite{rowshan-const_GRAND}, we studied how to constrain this single test error pattern generator to output the patterns satisfying one or multiple disjoint constraints. The aim was to avoid the computationally complex operation in \eqref{eq:check} in the pattern checking stage and replace it with a computationally simple partial pre-evaluation in the pattern generation stage. Towards this goal, we extracted multiple constraints from the original or manipulated parity check matrix such that the constraints cover disjoint sets of indices in $[1,n]$.  

In this section, we use the extracted constraints in \cite{rowshan-const_GRAND} and call the corresponding disjoint sets \emph{segments}. Furthermore, we employ multiple test error pattern generators associated with the segments to generate short patterns, named \emph{sub-patterns}, satisfying the constraint corresponding to the segments. Hence, unlike in \cite{rowshan-const_GRAND}, all the generated sub-patterns and the patterns resulting from the combinations of sub-patterns will satisfy all the constraints, and we do not discard any generated patterns. However, this advantage comes with the challenging problem of how to order test error patterns resulting from the combinations of sub-patterns. We will tackle this problem in the next section. In the remainder of this section, we define the segments and notation needed for the rest of the paper. 

Depending on the parity check matrix $\bH$ of the underlying code, we can have at least two segments. Denote the total number of segments by $p$ and the set of coordinates (or indices) of the coded symbols in the segment $j$ by $\cS_j$.  Any row $\bh_j,j\in[1,n-k]$, of matrix $\bH$  can partition the block code into two segments as follows: 
$$\cS_{j}=\supp(\bh_j),\;\;\;\;\;
\cS_{j}^\prime=[1,n]\backslash\supp(\bh_j).$$ 
Before further discussion, let us define explicitly a segment as follows:
\begin{definition}
    {\bf Error Sub-pattern}: A subset of coordinates in the test error pattern $\bhe$ corresponding to a segment is called an error sub-pattern. In other words, the error sub-pattern corresponding to segment $j$, denoted by $\cE_j$, is defined as
    \begin{equation}
        \cE_j=\cS_j\cap\supp(\be).
    \end{equation}
\end{definition}
The syndrome can give us some insight into the number of errors in each segment.
\begin{remark}
The corresponding element $s_j$ in syndrome $\bs=[s_{1} \; s_{2} \;\cdots \;s_{n-k}]$ determines the weight of the corresponding error sub-pattern $\cE_j$ as \cite{rowshan-const_GRAND}
    \begin{equation}\label{eq:even_odd}
        |\cE_j|=|\supp(\bh_j)\cap\supp(\be)|=
        \begin{dcases}
            \text{odd}  & s_j=1,\\ 
            \text{even} & s_j=0,
        \end{dcases}
    \end{equation}
    where the even number of errors includes no errors as well. However, the weight of the error sub-pattern corresponding to positions outside $\supp(\bh_j)$, i.e., $$|\big([1,n]\backslash\supp(\bh_j)\big)\cap\supp(\be)|\rightarrow\text{unknown},$$ can be either even or odd as the positions in $[1,n]\backslash\supp(\bh_j)$ are not involved in the parity constraint $\bh_j$.
\end{remark}

Depending on the parity check matrix $\bH$, we may be able to cover the positions in $[1,n]\backslash\supp(\bh_j)$ by one or more other rows in $\bH$ other than row $j$. This can be achieved by matrix manipulation of $\bH$, i.e., row operation, because the row space is not affected by elementary row operations on $\bH$ (resulting in $\bH'$) as the new system of linear equations represented in the matrix form $\bH'\cdot\bc=\bzero$ will have an unchanged solution set $\cC$. 

\begin{example}\label{ex:h1_subset_h2}
Suppose we have three rows of a parity check matrix and the associated syndrome bits as follows:
$$\bh_{j_1} = [1 \; 1  \; 1  \; 1 \; 0  \; 1  \;1  \;0], \;\;s_{j_1}=0,
$$
$$
\bh_{j_2} = [0 \; 1  \; 0  \; 1 \; 0  \; 0  \;1  \;0], \;\;s_{j_2}=1,$$
$$\bh_{j_3} = [0 \; 1  \; 0  \; 1 \; 1  \; 0  \;1  \;1], \;\;s_{j_3}=0.$$
From $\bh_{j_2}$, we can form two segments corresponding to the following disjoint index sets: 
$$\cS_{j_2}=\supp(\bh_{j_2})=\{2,4,7\},
$$
$$
\cS_{j_2}^\prime=[1,8]\backslash\supp(\bh_{j_2})=\{1,3,5,6,8\}.$$
From $s_{j_2}=1$, we understand that 
$$|\cS_{j_2}\cap\supp(\be)|\rightarrow\text{odd},\;\;\;\;\;
|\cS_{j_2}^\prime\cap\supp(\be)|\rightarrow\text{unknown}.$$
Here, unknown means the weight of error sub-pattern $\cE_{j_2}^\prime=\cS_{j_2}^\prime\cap\supp(\be)$ can be either even or odd. Hence, we have to generate all the sub-patterns, not constrained to odd or even sub-patterns only. Note that we can efficiently generate only odd or even sub-patterns as illustrated in Section \ref{sec:implement}, however in the case of no insight into the number of errors in the segment, we have to generate all possible sub-patterns for that specific segment.  

Now, by row operations on $\bh_{j_1}$ and $\bh_{j_3}$, we can get 
$$\bh_{j_1}^\prime = \bh_{j_1}\oplus\bh_{j_2} = [1 \; 0  \; 1  \; 0 \; 0  \; 1  \;0  \;0], \;\;s_{j_1}^\prime=1,$$
$$\bh_{j_2} = [0 \; 1  \; 0  \; 1 \; 0  \; 0  \;1  \;0], \;\;s_{j_2}=1,$$
$$\bh_{j_3}^\prime = \bh_{j_3}\oplus\bh_{j_2} = [0 \; 0  \; 0  \; 0 \; 1  \; 0  \;0  \;1], \;\;s_{j_3}^\prime=0,$$
where we can form three segments ($p=3$) with corresponding disjoint index sets 
$$\cS_{j_1}^\prime=\{1,3,6\},\;
\cS_{j_2}=\{2,4,7\},\;
\cS_{j_3}^\prime=\{5,8\},$$
from which we understand that the weight of error sub-patterns are as follows:
$$|\cS_{j_1}^\prime\cap\supp(\be)|\rightarrow\text{odd},\;\;
|\cS_{j_2}\cap\supp(\be)|\rightarrow\text{odd},\;\;
|\cS_{j_3}^\prime\cap\supp(\be)|\rightarrow\text{even}.$$

\end{example}

Now, we turn our focus to the possible reduction in complexity that segmentation can provide in terms of sorting complexity and membership-checking complexity.

{\bf Complexity of Sorting the Received Signals.} In all the variants of GRAND, the received signals should be sorted in ascending order of their absolute values. Let us take a bitonic network sorter with the total number of stages computed based on the sum of the arithmetic progression as \cite[Section V]{rowshan-lva}
\begin{align}\label{eq:num_stages}
    \Psi = \sum _{\psi =1}^{\log _2 n} \psi =\frac{1}{2}(\log _2 n)(1+\log _2 n). 
\end{align}
Observe that the reduction in $n$ can significantly reduce $\Psi$ as a measure of latency in a parallel implementation.  For example, given a code with length $n=64$ with $\Psi=21$. If it is segmented into two equal segments, then we get $\Psi=15$. 
Note that the total number of stages in \eqref{eq:num_stages} as a measure of time complexity (assuming that all nodes in every stage are processed simultaneously) is in the order of $O(\log _2^2 n)$. Clearly, by segmentation, $n$ reduces, and so does the time complexity.
Furthermore, assuming the use of the merge sort or quick sort algorithm with the complexity of $O\left(n\log_2 n\right)$ comparisons, we can similarly observe that reducing $n$ would significantly reduce the number of operations. For example, defining two equal segments can reduce $n\log_2 n$ to $2\times\left(\frac{n}{2}\log_2 \frac{n}{2}\right)=n\log_2 \frac{n}{2}$.

{\bf Average Number of Queries.}  As the reduction in the number of queries depends on the number of parity constraints, let us first see how many segments we can have.
\begin{remark}
    The maximum number of segments depends on the underlying code. However, the minimum number of segments is two as was shown in Example \ref{ex:h1_subset_h2} by considering either one or two parity check constraints. The latter gives a lower complexity because we get insight into both segments. Codes that have a well-structured parity check matrix, such as polar codes, can easily form more than two segments. 
\end{remark}
The reduction in the average complexity is also proportional to the reduction in the size of the search space, as was shown numerically in \cite{rowshan-const_GRAND}. The following lemma shows that the size of the search space reduces by a factor of two, and it depends on the total number of parity constraints. 
\begin{lemma}\label{lma:seach_space_size}
Suppose we have a parity check matrix $\bH$ in which there are $p$ rows of $\bh_{j},j=j_1,j_2,...,j_p$ with mutually disjoint index sets $\cS_j = \supp(\bh_{j})$ that define $p$ segments, then the size of the search space by these $p$ parity check equations is
    \begin{equation}\label{eq:search_space_size_less_than}
        \Omega(\bh_{j_1},..,\bh_{j_p})  = 2^{n-p}.
    \end{equation}
\end{lemma}
\begin{proof}
    Let us first take a row $\bh_j$ and $\cS_j=\supp(\bh_j)$. In this case, we only consider the error sequences satisfying $|\cS_j\cap\supp(\bhe)|\text{ mod }2= s_j$ in the search space. Then, the size of the constrained search space will be 
    \begin{equation}\label{eq:single_const}
        \Omega(\bh_j) = \sum_{\footnotesize\substack{\ell\in[0,|\cS_j|]:\\ \ell\text{ mod } 2 = s_j}} {|\cS_j| \choose \ell} \cdot 2^{n-|\cS_j|} = \frac{2^{|\cS_j|}}{2} \cdot 2^{n-|\cS_j|} = 2^{n-1}.
    \end{equation}
    Generalizing \eqref{eq:single_const} for $p$ constraints, we have
    \begin{multline*}
        \Big(\prod_{j=j_1}^{j_p}\sum_{\footnotesize\substack{\ell\in[0,|\cS_j|]:\\\ell\text{ mod } 2 = s_j}} {|\cS_j| \choose \ell}\Big) \cdot 2^{n-\sum_{j=j_1}^{j_p}|\cS_j|} =\\ \big(\prod_{j=j_1}^{j_p} 2^{|\cS_j|-1}\big) \cdot 2^{n-\sum_{j=j_1}^{j_p}|\cS_j|} = 2^{n-p}.
    \end{multline*}\vspace{-10pt}
    
\end{proof}

So far, we have defined the segments and the corresponding error sub-patterns. In \cite{rowshan-const_GRAND}, we provided an efficient scheme shown in Fig. \ref{fig:pre_eval} to evaluate the outputs of a single test error pattern generator of the ORBGRAND with respect to the segments' constraint in \eqref{eq:even_odd} before checking the codebook membership by \eqref{eq:check}. 
However, this paper suggests using multiple error pattern generators shown in Fig. \ref{fig:sub_patterns} that only produce valid sub-patterns simultaneously for the associated segments. Hence, the pre-evaluation stage in Fig. \ref{fig:pre_eval} is no longer required. This sub-pattern-based approach is discussed in the next section in detail. 

\begin{figure}[ht]
    \centering
    \includegraphics[width=0.9\columnwidth]{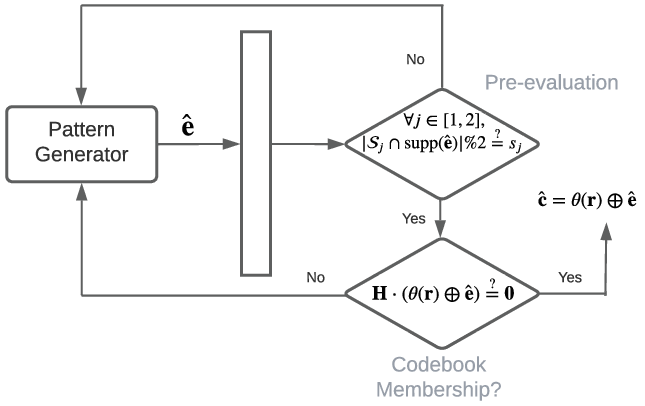}
    \caption{The error pattern generation process with pre-evaluation based on two constraints in ``constrained GRAND" \cite{rowshan-const_GRAND}.}
    \label{fig:pre_eval}
\end{figure}

\begin{figure}[ht]
    \centering
    \includegraphics[width=0.9\columnwidth]{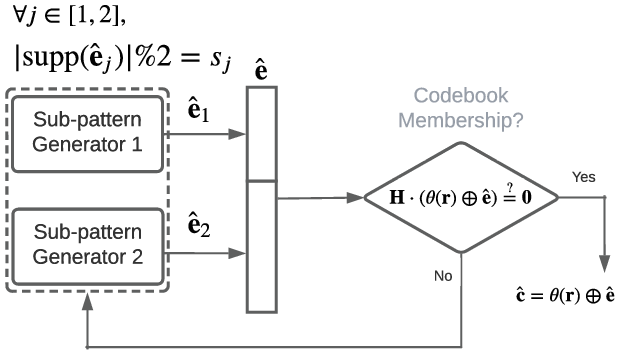}
    \caption{The proposed error pattern generation approach based on two sub-patterns in ``segmented GRAND".}
    \label{fig:sub_patterns}
\end{figure}


\section{Combining Sub-patterns in Near-ML Order} \label{sec:comb_patterns}
A challenging problem in handling sub-patterns is combining them in an order near the ML order. In ORBGRAND, the logistic weight $w_L$ is used as a guide to generate error patterns in a near-ML order, i.e., the logistic $w_L$ is assumed to be proportional to the increase in distance, $d^{(+)}$ from the received sequence as shown in \eqref{eq:wL_dist}. As discussed in the previous section, we eventually want to generate sub-patterns for the segments of the underlying code and then combine them. 
Every segment $j$ uses its own logistic weight $w_L^{(j)}$ to generate its sub-patterns. Hence, within each segment, the patterns are generated in a near-ML order. 
However, we do not know how to combine the sub-patterns from different segments in order to generate the entire pattern in a near-ML order. The trivial way would be generating a set of entire patterns by considering all the possible combinations of the sub-patterns (probably in batches due to the limitation of resources), computing their SEDs, and then sorting them in the ascending order of SED. This method is not of our interest because we need to store many patterns and sort them frequently, similar to what we do in soft-GRAND. 
In this section, we propose an approach based on a logistic weight $w_L$, to preserve the near-ML order in the ORBGRAND, in which we assign sub-weights $w_L^{(j)},j\in[1,p]$ to $p$ segments such that 
\begin{equation}\label{eq:sub_weights}
    w_L = \sum_{j=1}^p w_L^{(j)}.
\end{equation}
Observe that the combined sub-patterns will still have the same $w_L$ for any set of $[w_L^{(1)}\;\;w_L^{(2)}\cdots w_L^{(p)}]$ that satisfies \eqref{eq:sub_weights}. 
Note that that in order for the least reliable bits in each segment $\mathcal{S}_j$ to contribute the same amount of logistic weight to the total logistic weight $w_L$, the number of elements in a segment should be large (e.g., $n = 128$ and two segments of 64 bits each, in which case the rank ordered reliability for each segment will be similar) or the number of segment should be small (e.g., $p=2$). 
Now the question is how to get all such sub-weight vectors $[w_L^{(1)}\;\;w_L^{(2)}\cdots w_L^{(p)}]$. It turns out that by modification of integer partitioning defined in Definition \ref{def:int_part}, we can obtain all such sub-weights. 
The difference between the integer partitions in Definition \ref{def:int_part} and what we need for sub-weights are as follows: 
1) The integer partitions do not need to be distinct (repetition is allowed). That is, two or more segments can have identical sub-weights, 2) the permutation of partitions is allowed, 3) the number of integer partitions (a.k.a part size) is fixed and is equal to the number of segments, and 4) the integer zero is conditionally allowed, i.e., one or more partitions can take zero value given the syndrome element corresponding to the segment is $s_j=0$. 

After obtaining the sub-weights, we can use the integer partitions in Definition \ref{def:int_part} to get the sub-pattern(s). Hence, we have two levels of integer partitioning in the proposed approach. These two levels are illustrated in Fig. \ref{fig:two-llvl-part}. The rest of this section is dedicated to giving the details of this approach starting with some examples for the first level of partitioning and then some definitions and a proposition on how to get all the valid sub-weights for the segments in an efficient way. 

\begin{figure}[ht]
    \centering
    \includegraphics[width=1\columnwidth]{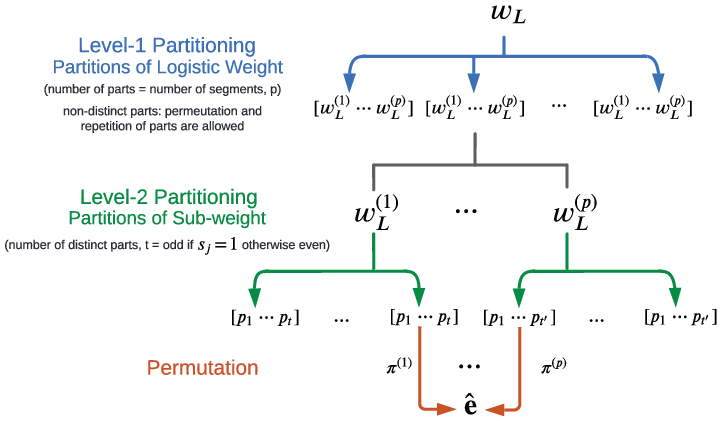}
    \caption{Two-level integer partitioning to generate error patterns for $p$ segments. Note that we have $j\in[1,p]$ and $t,t'$ are the number of parts (odd, even, or arbitrary when we don't have $s_j$ for the corresponding segment such as segment $\cS_{j_2}^\prime$ in Example \ref{ex:h1_subset_h2}).}
    \label{fig:two-llvl-part}
\end{figure}

\begin{example}
    Suppose the current logistic weigh is $w_L=5$ and the codeword is divided into three segments, $p=3$, with the corresponding syndrome elements $s_{j_1}=0$, $s_{j_2}=1$ and $s_{j_3}=1$. That is, the weights of the sub-patterns corresponding to the segments are [even, odd, odd], respectively. To generate the sub-patterns for this $w_L$, the logistic weights of the segments by the first level of integer partitioning are chosen as 
    $$[0\;1\;4], [0\;2\;3], [0\;3\;2], [0\;4\;1], [3\;1\;1].$$
    Observe that the sum of the segment weights is 5 while there are repetitions of weights in $[3\;1\;1]$, permutation of the weights in $[0\;2\;3]$ and $[0\;3\;2]$, and zero weight for the segment with $s_{j_1}=0$ to allow considering no errors for segment $j_1$, i.e., empty sub-pattern. We will discuss later the other details of the sub-pattern generation shown in this example.
\end{example}
Note that the sub-patterns are generated based on the logistic weight of the segments provided in the example above at the second level of integer partitioning where the parts are distinct integers, in a manner employed in ORBGRAND. 

\begin{example}\label{ex:subpatterns-shared}
    Suppose we have three segments and $[w_L^{(1)}\;w_L^{(2)}\;w_L^{(3)}]=[0\;3\;5]$ for $w_L=8$. The integer partitioning of 3 and 5 with distinct parts results in $[1\;2]$ for $w_L^{(2)}=3$, and $[1\;4]$ and $[2\;3]$ for $w_L^{(3)}=5$. Therefore, there are $1\times2\times3=6$ sub-patterns as follows:
    $$[\;]+[3]+[5],\;\;\;[\;]+[1\;2]+[5],$$ $$[\;]+[3]+[1\;4],\;\;\;[\;]+[1\;2]+[1\;4],$$ $$[\;]+[3]+[2\;3],\;\;\;[\;]+[1\;2]+[2\;3].$$ 
     Observe that the sub-patterns are shared among several error patterns. This allows us to reuse the sub-patterns in generating new patterns, significantly reducing the average complexity of pattern generation. 
\end{example}
{\bf Local permutation.} The integers in the aforementioned sub-patterns refer to the relative position of the symbols in the segments, locally ordered with respect to their reliability. Hence, we need to use a local permutation $\pi^{(j)}(\cdot)$ for every segment $j$, unlike the ORBGRAND where we have only one permutation function $\dot{\pi}(\cdot)$ as discussed in Section \ref{sec:prelim}. 
    The operator ``+" denotes the concatenation of the sub-patterns. These patterns can be checked in an arbitrary order as long as they belong to the same $w_L$.
The local permutation $\dot{\pi}^{(j)}(\cdot)$ maps a local index in $[1,|\cS_j|]$ belonging to segment $j$ to the overall index in $[1,n]$ as  
\begin{equation}\label{eq:elem-perm}
    \dot{\pi}^{(j)} : \{1,2,\dots,|\cS_j|\} \rightarrow \cS_j.
\end{equation}

From Definition \ref{def:int_part}, we can define a $\bz_j$ as a binary vector of length $|\cS_j|$ in which $z_{j,i}=1$ where $i\in\cI\subset[1,|\cS_j|]$ for $w_L^{(j)}=\sum_{i\in\cI} i$. Then, 
the element-wise permutation from $w_L^{(j)},j=1,\dots,p$ can be used to flip the relevant positions in an all-zero binary vector with length $n$ to obtain the error pattern vector $\be$, as shown in Fig. \ref{fig:two-llvl-part}.

\begin{example}\label{ex:local_permut}
    Suppose we have the received sequence $\br=[0.5,-1.2,0.8,1.8,-1,-0.2,0.7,-0.9]$ similar to Example \ref{ex:permut}. We use the segments defined in Example \ref{ex:h1_subset_h2} as 
        $$\cS_{j_2}=\{2,4,7\},\;\;\;\cS_{j_2}^\prime=\{1,3,5,6,8\}.$$
    Now, the local permutation function based on $|r_i|,i\in[1,8]$ in ascending order can be obtained as follows:
    \[
        \dot{\pi}^{(j_2)} : [1,2,3]\rightarrow[7,2,4], \;\;\;\;
        \dot{\pi}^{\prime(j_2)} : [1,2,3,4,5]\rightarrow[6,1,3,8,5]
    \]
\end{example}
Now, let us define an efficient framework for error pattern generation based on sub-patterns that plays the role of guidelines to generate valid sub-patterns only. This framework consists of bases for the formation of error patterns and a minimum logistic weight that each base can take.   We begin with defining the bases with respect to the syndrome elements as follows:

\begin{definition}
    \textbf{Error Pattern Bases}: A base for the error patterns, denoted by $[f_1\,f_2\,\dots\,f_p]$ for $p$ segments, determines the segments contributing their sub-patterns to the error patterns given by logistic weight $w_L$ as 
    \begin{equation}
        w_L = \sum_{j=1}^p f_j\cdot w_L^{(j)}
    \end{equation}
    where $f_j$ can get the following values: 
    \begin{equation}\label{eq:}
    f_j=
    \begin{dcases}
        \{0,1\} & s_j=0,\\
        \{1\} & s_j=1.\\
    \end{dcases}
    \end{equation}
    The segments with $f_j=0$ are called \emph{frozen segments} where the sub-pattern contributed by segment $j$ is empty. The total number of bases is $\prod_{j=1}^{p} 2^{1-s_j}$ that can be between 1 and $2^p$ depending on $s_j,j\in[1,p]$.   
\end{definition}
Note that when $s_j=0$ for segment $j$, this segment might be error-free. That is the reason why we have error pattern bases excluding the sub-patterns of such segments by setting $f_j=0$. Moreover, when we have $s_j=0$ and $f_j=1$, since the segment $j$ can have sub-patterns with an even weight and the smallest even number of parts is 2, we need to have $w_L^{(j)}\geq 3$ as $3=1+2$ gives the first two most probable erroneous positions. That is, the first error pattern $\bz$ for this segment will be $z_1=z_2=1$ and $z_i=0,i\geq3$ or $\bz=[1\;1\;0\cdots0]$.
On the contrary, we necessarily need $f_j=1$ and $w_L^{(j)}\geq 1$ when $s_j=1$. That is, we cannot have an empty sub-pattern for such segment $j$ in this case.

\begin{proposition}\label{prp:min_weight}
    Given the segments' syndrome $[s_1\,s_2\,\cdots\,s_p]$ and the pattern base $\bs=[f_1\,f_2\,\dots\,f_p]$ for $p$ segments, the minimum $w_L$ that every pattern base can give is 
    \begin{equation}\label{eq:wL_min_overall}
        \underline{w}_L(\bs) = \sum_{j=1}^p f_j\cdot \underline{w}_L^{(j)}(s_j),
    \end{equation}
    where $\underline{w}_L^{(j)}(s_j)$ is 
    \begin{equation}\label{eq:wL_min}
        \underline{w}_L^{(j)}(s_j) = 3 - 2s_j.
    \end{equation}
    Thus, the overall logistic weight $ w_L(\bs)$ and sub-weights $w_L^{(j)}(s_j),j\in[1,p]$ must satisfy
    \begin{equation}\label{eq:wt_constraints}
        w_L(\bs) \geq \underline{w}_L(\bs) \text{ and }  w_L^{(j)}(s_j) \geq \underline{w}_L^{(j)}(s_j).
    \end{equation}
\end{proposition}
\begin{proof}
    Equation \eqref{eq:wL_min} follows from function $\underline{w}_L^{(j)}(s_j):\{0,1\}\rightarrow\{3, 1\}$ as discussed earlier, which maps the minimum non-zero $w_L^{(j)}$ to 3 when $s_j=0$ and maps to 1 when $s_j=1$. Then, Equation \eqref{eq:wL_min_overall} clearly holds for the minimum of overall logistic weight which is denoted by $\underline{w}_L(\bs)$.
\end{proof}
Observe that the base patterns are used to efficiently enforce the minimum weight constraints in \eqref{eq:wt_constraints}. The importance of the base patterns is realized when we recall that the level-1 integer partitioning allows permutation and repetition of parts (here, sub-weights).
\begin{example}\label{ex:4}
    Given $s_1=0,s_2=1$ and $s_3=1$, we would have $2\times 1\times 1=2$ error pattern bases $[f_1\;f_2\;f_3]$ and their minimum weights/sub-weights as follows:
    $$[f_1\;f_2\;f_3]=[0\;1\;1],\underline{w}_L=2,[\underline{w}_L^{(1)}=0\;\;\underline{w}_L^{(2)}=1\;\;\underline{w}_L^{(3)}=1],$$
    $$[f_1\;f_2\;f_3]=[1\;1\;1],\underline{w}_L=5,[\underline{w}_L^{(1)}=3\;\;\underline{w}_L^{(2)}=1\;\;\underline{w}_L^{(3)}=1].$$
    Now, for $w_L=4$, the sub-weights  $[w_L^{(1)}\;w_L^{(2)}\;w_L^{(3)}]$ are $[0\;1\;3],[0\;2\;2],$ and $[0\;3\;1]$. As can be seen, $w_L^{(1)}=0$, i.e., segment 1 is frozen, and all the sub-weights were generated with the pattern base $[0\;1\;1]$. However, for $w_L=5$, the sub-weights are $[0\;1\;4],[0\;2\;3],[0\;3\;2],[0\;4\;1],$ and $[3\;1\;1]$ where the last one is based on the pattern base $[1\;1\;1]$ (note that $\underline{w}_L=5$ for this base). 
\end{example}
Following the example above, we define our tailored integer partitioning scheme for combining the sub-patterns.
\begin{definition}\label{def:}
    \textbf{Logistic Weight and Sub-weights}: 
    Suppose we have a block code with $p$ segments. The overall logistic weight $w_L$ can be distributed among segments by sub-weights  $w_L^{(j)} = \kappa_j + c_j$ as 
    \begin{equation}
        w_L = \sum_{j=1}^p f_j \cdot w_L^{(j)} = \sum_{j=1}^p f_j (\kappa_j + c_j),
    \end{equation}
where $\kappa_j\geq \underline{w}_L^{(j)}$ is the initial value for $w_L^{(j)}$ 
and $c_j\geq 0$ is the increments to get larger $w_L^{(j)}$.  
\end{definition}

\begin{example}\label{ex:two_seg_two_lev}
    Given $s_1=1$ and $s_2=0$, we would have $1\times 2=2$ error pattern bases $[f_1\;f_2]$ as follows:
    $$[f_1\;f_2]=[1\;0],\underline{w}_L=1, [\underline{w}_L^{(1)}=1\;\;w_L^{(2)}=0],$$
    $$[f_1\;f_2]=[1\;1],\underline{w}_L=4, [\underline{w}_L^{(1)}=1\;\;\underline{w}_L^{(2)}=3].$$
    As Fig. \ref{fig:2seg_level1_parts} shows, segment 2 is frozen up to $w_L=3$ and all sub-weights are generated by base $[1\;0]$ at level-1 partitioning. Hence no error pattern is allocated to this segment for $1\leq w_L\leq 3$. 
    Note that the all-zero error pattern is not valid in this case, i.e., $w_L>0$. 
    \begin{figure}[ht]
        \centering
        \includegraphics[width=0.7\columnwidth]{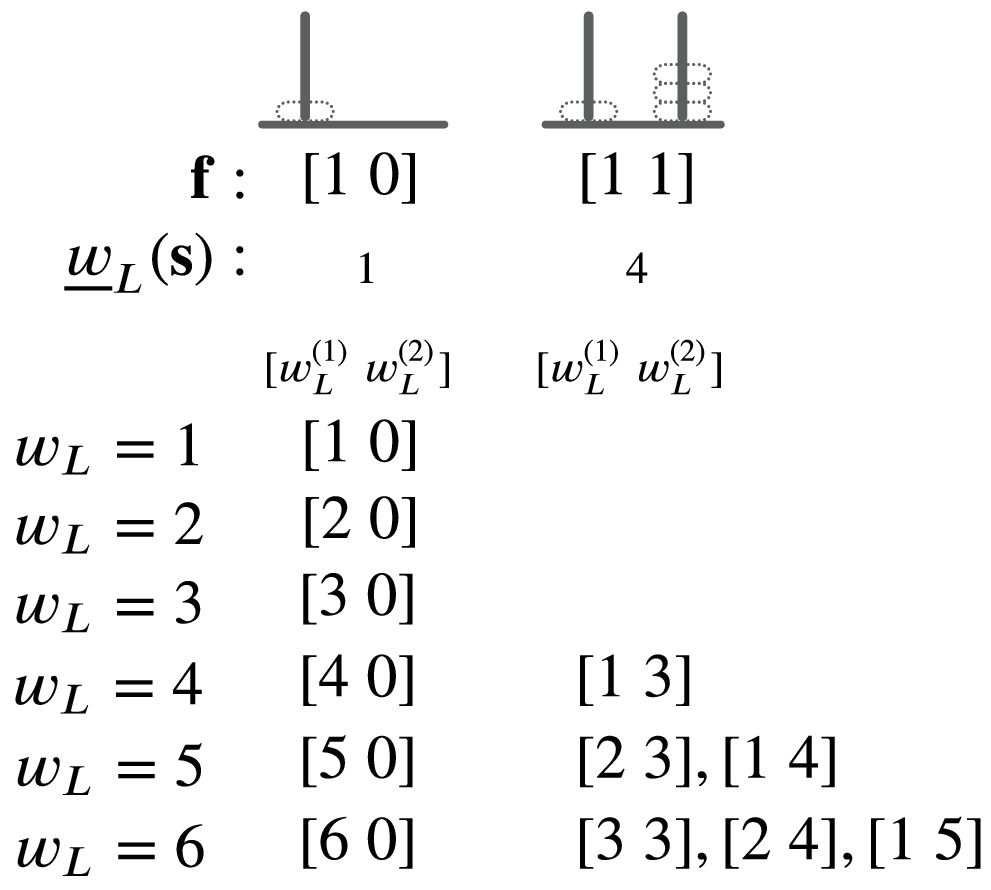}
        \caption{The sub-weights generated based on $\bs=[s_1=1\;s_2=0]$ for two-segment based GRAND. For $w_L=1,2,3$, the base $\bf=[1\;0]$ is activated only because  the base $\bf=[1\;1]$ has $\underline{w}_L=4$. We have both bases activated for $w_L\geq4$.}
        \label{fig:2seg_level1_parts}
    \end{figure}
    Furthermore, Fig. \ref{fig:two_level_part} shows the two levels of partitioning specifically for $w_L=6$ when the partitions $[w_L^{(1)},w_L^{(2)}]=[2\;4]$ is selected in the first level. Following the permutation functions in Example \ref{ex:local_permut}, the error pattern vector $\be$ is given as well.
    \begin{figure}[h]
        \centering
        \includegraphics[width=0.9\columnwidth]{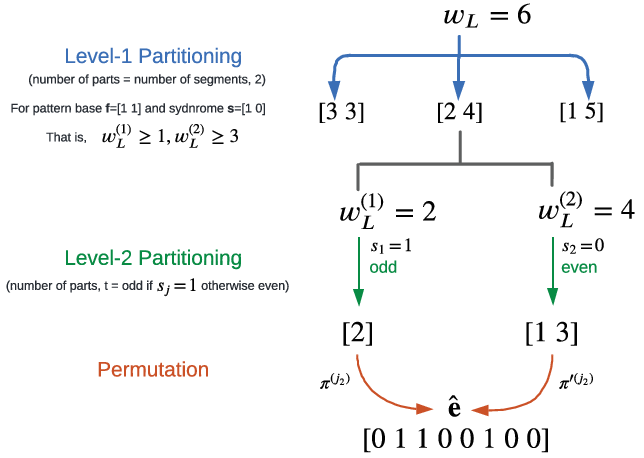}
        \caption{An example of two-level error pattern generation based on sub-patterns when $\bs=[s_1=1\;s_2=0]$. Note that $n1$ and $n2$ are the lengths of segments 1 and 2.}
        \label{fig:two_level_part}
        \vspace{-5pt}
    \end{figure}
\end{example}
The idea of splitting the logistic weight into sub-weights for the segments is based on the assumption that the least reliable symbols are almost evenly distributed among the segments. The statistical results for 15000 transmissions of eBCH(128,106) codewords over AWGN channel show that this assumption is actually realistic. Fig. \ref{fig:low-reli_distrib} shows the distribution of 64 least reliable symbols between two 64-length segments, by locating and counting them in the segments for each transmitted codeword. The mean and standard deviation of the bell-shaped histogram for each segment is 32 and 2.85, respectively. Moreover, as the additive noise follows Gaussian distribution and it is independent and identically distributed among the symbols, these results were expected. 
\begin{figure}[h]
    \centering
    \includegraphics[width=0.8\columnwidth]{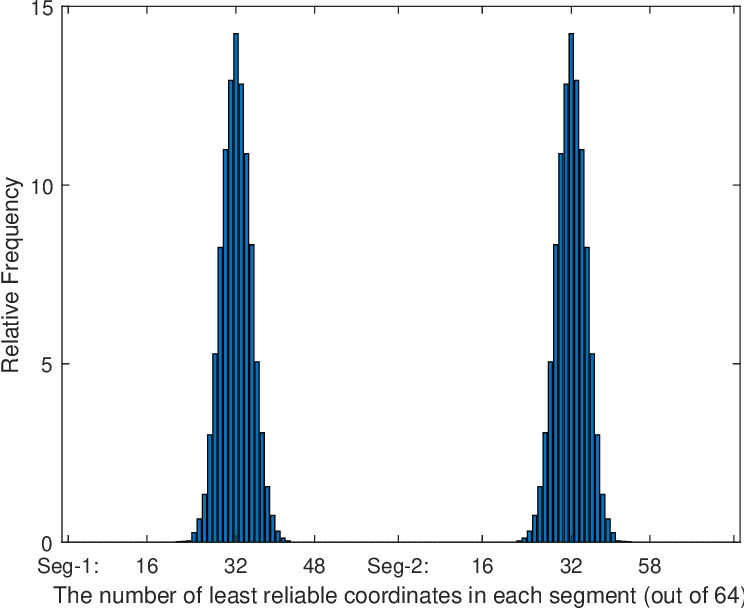}
    \caption{Distribution of the 64 least reliable coordinates between two segments for the 15000 independent transmissions of eBCH(128,106) codewords.} 
    \label{fig:low-reli_distrib}
    \vspace{-5pt}
\end{figure}

Now, let us look at a realistic example comparing the ORBGRAND with Segmented ORBGRAND in terms of searching for a valid error pattern.
\begin{example}\label{ex:err_pattern_comparison}
    Suppose a codeword of eBCH code (64,45) is transmitted over an AWGN channel and the hard decision on the received sequence leads to three erroneous bits at coordinates of $\supp(\be)=[1\;15\;23]$. One employs ORBGRAND to find these coordinates. This goal is achieved after 57 attempts sweeping through logistic weights $w_L=1\rightarrow 12$. Fig. \ref{fig:ex_queries_seq} illustrates the Euclidean distance of all queries. 
    \begin{center}
    \footnotesize
        \begin{tabular}{|c |c |c| c|| c| c| c| c|}
         \hline
         & $w_L$ & $\supp(\bhe)$ & $d_E^2$ & & $w_L$ & $\supp(\bhe)$ & $d_E^2$ \\ 
         \hline\hline
         1 & 1 & $\{23\}$ & 20.58 & 8 & 5 & $\{39,1\}$ & 21.49\\ 
         \hline
         2 & 2 & $\{1\}$ & 20.80 & 9 & 5 & $\{42\}$ & 21.68\\ 
         \hline
         3 & 3 & $\{23,1\}$ & 20.93 & $\cdots$ & $\cdots$ & $\cdots$ & $\cdots$\\ 
         \hline
         4 & 3 & $\{39\}$ & 21.14 & 54 & 11 & $\{42,33\}$ & 23.57\\ 
         \hline
         5 & 4 & $\{39,23\}$ & 21.27 & 55 & 12 & $\{36,23\}$ & 22.95\\ 
         \hline
         6 & 4 & $\{9\}$ & 21.35 & 56 & 12 & $\{50\}$ & 22.98\\ 
         \hline
         7 & 5 & $\{23,9\}$ & 21.48 & 57 & 12 & $\{23,15,1\}$ & 23.10\\ 
         \hline
        \end{tabular}
    \end{center}
    Now, if one divides the codeword into two equal-length segments with coordinates in $\cS_1,\cS_2$ based on two constraints, it turns out that segments 1 and 2 have odd and even numbers of errors since $s_1=1$ and $s_2=0$. The proposed Segmented ORBGRAND can find the error coordinates in only 7 queries as illustrated in the table below. 
    \begin{center}
    \footnotesize
        \begin{tabular}{|c |c |c| c| c| c|}
         \hline
         & $w_L$ & $[w_L^{(1)}\;w_L^{(2)}]$ & $\supp(\bhe)\cap\cS_1$ & $\supp(\bhe)\cap\cS_2$ & $d_E^2$ \\ 
         \hline\hline
         1 & 1 & $[1\; 0]$ & $\{9\}$ & $\{\}$ & 21.35\\ 
         \hline
         2 & 2 & $[2\; 0]$ & $\{15\}$ & $\{\}$ & 22.54\\ 
         \hline
         3 & 3 & $[3\; 0]$ & $\{8\}$ & $\{\}$ & 22.71\\ 
         \hline
         4 & 4 & $[4\; 0]$ & $\{50\}$ & $\{\}$ & 22.98\\ 
         \hline
         5 & 4 & $[1\; 3]$ & $\{9\}$ & $\{23,1\}$ & 21.83\\ 
         \hline
         6 & 5 & $[5\; 0]$ & $\{3\}$ & $\{\}$ & 23.06\\ 
         \hline
         7 & 5 & $[2\; 3]$ & $\{15\}$ & $\{23,1\}$ & 23.10\\ 
         \hline
        \end{tabular}
    \end{center}
    The patterns found by Segmented ORBGRAND are circled in Fig. \ref{fig:ex_queries_seq}. As can be seen, by segmentation, we can avoid checking many invalid error patterns.  Note that in this example, since the logistic weight of each error location is smaller than the number of redundant bits in the systematic form, all error positions belong to the set of least reliable bits in the order statistics decoding. Therefore, this error combination can be corrected on the first attempt by an $i$-order OSD decoding. However, we cannot be certain that this candidate codeword is the closest to the received sequence, necessitating further attempts and comparisons.
    \begin{figure}[h]
        \centering
        \includegraphics[width=0.8\columnwidth]{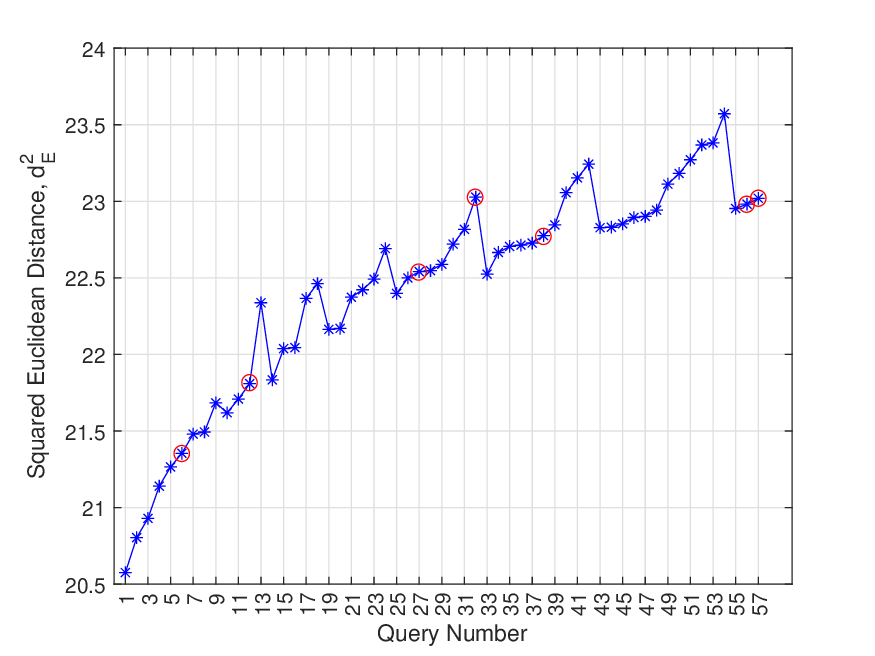}
        \caption{The squared Euclidean distance of queries up to the first codeword in ORBGRAND. The red circles indicate the queries performed by Segmented ORBGRAND with an order different from ORBGRAND. Note that since the metric is not \emph{monotonically increasing}, i.e., not always increasing or remaining constant, it doesn't give the ML order. } 
        \label{fig:ex_queries_seq}
        \vspace{-5pt}
    \end{figure}
\end{example}
\begin{remark}
    The logistic weight $w_L$ in ORBGRAND differs from that in segmented ORBGRAND because they generate different error patterns. This difference arises from the disagreement between the permuted integer parts of one segment in ORBGRAND and multiple segments in segmented ORBGRAND, unless the symbol reliability yields the following permutations for the cases of one segment and two segments (or equivalently for more segments):
    \begin{equation}
        \begin{cases}
        \dot{\pi}(i) = \dot{\pi}^{(1)}(i) & \forall i\in[1,|\mathcal{S}_1|],\\ 
        \dot{\pi}(|\mathcal{S}_1|+i) = \dot{\pi}^{(2)}(i) & \forall i\in[1,|\mathcal{S}_2|],\\ 
        \end{cases}
    \end{equation}
    which represents  just one permutation out of all $n!$ possible permutations. 
    In Example \ref{ex:err_pattern_comparison}, $w_L=5$ for segmented ORBGRAND gives the flipping coordinates 
    $$\dot{\pi}^{(1)}(\{2\})\cup \dot{\pi}^{(2)}(\{1,2\})=\{15\}\cup\{23,1\},$$
    whereas achieving these flipping coordinates in ORBGRAND requires $w_L=12$ as
    $$\dot{\pi}(\{1,2,9\})=\{23,1,15\}.$$
    Notice that due to the discrepancy in permutations, coordinate 15, which is the 9th least reliable coordinate in the entire sequence, becomes the 2nd least reliable coordinate in Segment 1. Thus, this flipping pattern is reached faster by segmentation, utilizing a smaller $w_L$. 
\end{remark}
\begin{remark}
    We utilize multiple sub-pattern generators, and it might appear that the complexity would increase linearly with the number of segments for each generated pattern. However, would like to point out that this is not the case. The reason is that the weight of the constituent sub-patterns within each segment proportionally reduces due to their relatively smaller sub-weights compared to the overall weight. Furthermore, we only generate a subset of the error patterns that ORBGRAND produces. Consequently, as demonstrated by the number of operations in Fig.~\ref{fig:op-ebch-compar}, 
the complexity significantly decreases with the segmentation approach.
\end{remark}
We can combine the sub-patterns generated by $w_L^{(j)},j=1,\dots,p$, in an arbitrary order as the sub-weights of all the combinations are summed up to $w_L$.

\section{Tuning Sub-weights for Unequal Distribution of Errors among Segments}\label{sec:tuning}
In the previous section, we suggested initializing the parameter $\kappa_j$ by $\underline{w}_L^{(j)}\in\{1,3\}$ depending on the value of $s_j\in\{1,0\}$. Although we are considering the AWGN channel where the random noise added to every symbol is independent and identically distributed (i.i.d.), there is a possibility that the distribution of errors is significantly unbalanced, that is, the weight of the error vector in one segment is quite larger than the other one(s). We can get statistical insight into this distribution by counting the low-reliability symbols (or small $|r_i|$) in each segment, denoted by $a_j$. Then, we adjust the initialization of $\kappa_j$s and make them proportional to the number of symbol positions in the segment with $|r_i|<\epsilon$ where $\epsilon$ is an arbitrary threshold for low-reliability symbols. This will account for the unequal distribution of errors among the segments. Suppose the expected number of errors in a segment with length $L$ is 
\begin{equation}
    \mu_e^{(j)} = L\cdot 2P(|r|<\epsilon), 
\end{equation}
where the probability $Pr(\cdot)$ follows the Gaussian distribution with mean $1$ and noise variance $\sigma_n^2$. Then, we can adjust $k_j$ as 
\begin{equation}
    k_j= \underline{w}_L^{(j)}+\Big\lceil \frac{\max\{a_j\}-a_j}{\rho\cdot\mu_e^{(j)}}\Big\rceil,
\end{equation}
where $\rho\cdot\mu_e^{(j)}$ is used for normalization of the relative difference with the number of low-reliability symbols in the segments. The parameter $\rho\leq 0 $ can be adjusted to get a better result. We denote the second term by $\tau_j=\Big\lceil \frac{\max\{a_j\}-a_j}{\rho\cdot\mu_e^{(j)}}\Big\rceil$. 
Note that the offset $\tau_j$ is only for the initialization stage and we should consider it when we conduct integer partitioning of $w_L^{(j)}$ by subtracting the offset from the segment weight, i.e., $w_L^{(j)}-\tau_j$. Although these adjustments (addition and subtraction of $\tau_j$) seem redundant and ineffective, they will postpone the generation of large-weight patterns for the segment(s) with small $a_j$, and hence we will get a different order of patterns that may result in a fewer number of queries for finding a valid codeword. Let us have a look at an example.
\begin{example}
    Let us consider two segments with $L=32$ elements, the corresponding syndrome element $s_1=1,s_2=0$, threshold  $\epsilon=0.2$ and $\mu_e=L\cdot 2Pr(|r|<\epsilon)=8$. We realize that there are 11 and 3 elements in segments 1 and 2, respectively, satisfying $|r_i|<\epsilon$. Having fewer low-reliable positions than the expected number (i.e., $3<\mu_e$) implies that the possibility of facing no errors in segment 2 is larger than having at least 2 errors (recall that the Hamming weight of error sub-pattern for this segment should be even due to $s=0$). Therefore, in level-1 partitioning, we can increase the initial sub-weight for this segment from $\kappa_2=\underline{w}_L^{(2)}=3$ to $\kappa_2=\underline{w}_L^{(2)}+\tau_2=5$ by $\tau_2=2$ assuming $\rho=1/2$. This increase will delay generating sub-patterns with base $[1\;\;1]$ from $w_L=4$ to $w_L=6$ in Example \ref{ex:two_seg_two_lev}. This prioritizes checking all sub-patterns with sub-weights $[4\;\;0]$ and $[5\;\;0]$ hoping that we find the correct error pattern faster by postponing the less likely error patterns to a later time. Nevertheless, in level-2 partitioning when we want to generate the sub-patterns with sub-weight $\kappa_2=5$ and base $[1\;\;1]$, we should subtract $\tau_2$ from $\kappa_2=\geq5$; otherwise, we will miss the error patterns with smaller sub-weights, i.e., $w_L^{(2)}=3,4$.
\end{example}
The numerical evaluation of this technique for eBCH(128,106) with two segments and $\rho=0.3,\epsilon=0.2$ shown in the table below reveals a slight reduction in the average queries while the BLER remains almost unchanged. 
The reduction in queries can be attributed to cases where there is an imbalance in the distribution of low-reliability symbols across segments. However, a significant imbalance between segments does not necessarily imply a significant imbalance in the distribution of erroneous coordinates. Consequently, tuning techniques in such scenarios may necessitate relatively larger queries, leading to only a slight reduction overall. These results demonstrate that the original segmented ORBGRAND without tuning overhead is good enough despite not considering the reliability imbalance between the segments. The reason comes from the imperfection of the reliability metric and complexity averaging over all received sequences.
\begin{center}
\footnotesize
\begin{tabular}{c|ccccc} 
 \hline
$E_b/N_0$ & 3.5 & 4 & 4.5 & 5 & 5.5  \\ 
 \hline
without tuning & 30685 & 8358 & 1750 & 315 & 54  \\ 
 \hline
with tuning & 30492 & 8110 & 1661 & 291 & 48  \\ 
 \hline
\end{tabular}
\end{center}



\section{Complexity Analysis}\label{sec:complx}
In this section, we discuss the expected reduction in the complexity (the average number of queries) of the proposed scheme. The total size of the search space is considered $2^n$ where we have $2^k$ valid codewords. According to Theorem 2 in \cite{duffy-tit}, the distribution of the number of guesses for a non-transmitted codeword is almost exponential, with the rate $2^{n(1-R)}$ as the length $n$ of the binary codeword increases. Consequently, the complexity required to achieve ML performance of any GRAND algorithm is a function of redundancy, $n-k$, which is of the order of $2^{n-k}$ queries. Alternatively, one can consider the geometric distribution (as exponential distribution is a continuous analogue of the geometric distribution) where the random variable $X$ is defined as the number of failures until the first success, i.e., finding the first valid codeword. 
\begin{figure}
    \centering
    \includegraphics[width=0.8\columnwidth]{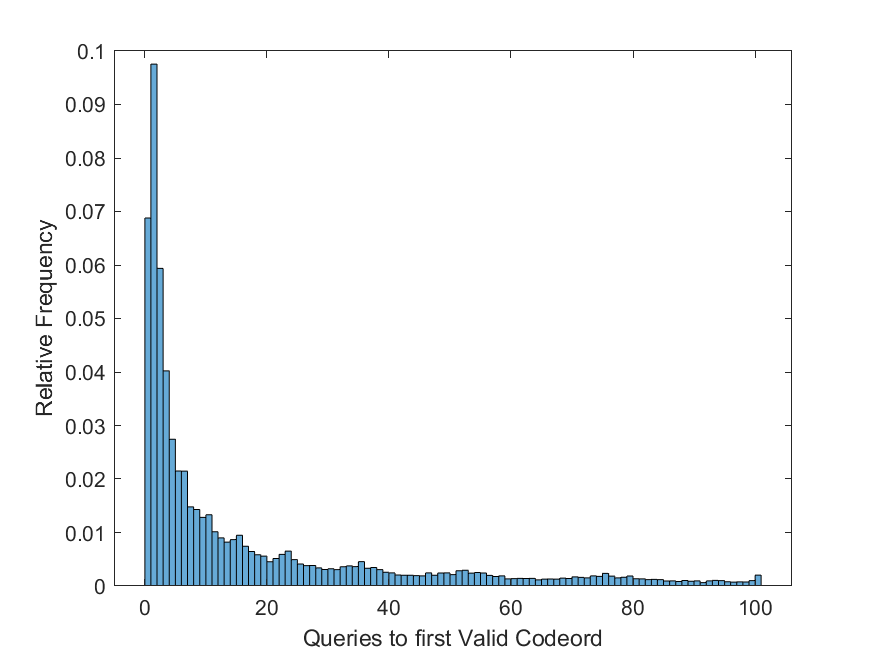}
    \caption{Relative frequency ($\approx p_i$) of the queries ($x_i,i\in[1,100]$) to the first valid codeword under ORBGRAND order for decoding 80,000 eBCH(64,45) codewords at $E_b/N_0=4$ dB (56\% of all decoding operations required less than 101 queries).}
    \label{fig:geo_queries}
\end{figure}
Regardless of the asymptotic distribution, the expected value $E(X)=\sum_{i}x_i\cdot p_i$, with $x_i=i$ and $p_i$ as the probability of finding a valid codeword at the $i$-th query, is a measure of the central tendency of a probability distribution, and it is calculated as the weighted average of all possible outcomes $x_i$, where the weights are the probabilities of each outcome $p_i$. Then, the probability of finding the first valid codeword after $m\geq1$ queries is $P(X=m)\approx\prod_{i=1}^{m-1}(1-p_i)p_m$ where we may not have $p_i=p_j$ for any $i\not = j$. 
The probability of finding a valid codeword $p_i$ changes by SNR and by the size of the search space. 
The reduction in the sample space increases the probability of the outcomes, $p_i$. As the relative frequency of small $x_i$ in Fig. \ref{fig:geo_queries} (or the probability of small X in exponential and geometric distribution) is considerably larger than large ones, the expected value is shifted towards a smaller value, i.e., the expected value of queries will decrease.  

Now, let us consider the scenario where the search for a valid codeword is abandoned after $b$ queries. In this scenario, similar to the queries without abandonment, we have a reduction in complexity. Moreover, the abandonment threshold $b$ limits the scope of queries leading to potential decoding failure in ORBGRAND. Fig. \ref{fig:success_query} (a) illustrates the failure due to the limited scope of the search. As can be seen, the reduction of search space in (b) helps the valid codeword falling into the scope of queries with threshold $b$. Hence, the reduction in the search space of segmented ORBGRAND is equivalent to increasing the threshold $b$ of ORBGRAND under abandonment. 
As the maximum query in practice could be a bottleneck of the system and therefore it is important to evaluate the decoding performance and complexity under the abandonment scenario, we consider these two scenarios in the evaluation of segmented ORBGRAND in section \ref{sec:num_results}. 
\begin{figure}
    \centering
    \includegraphics[width=0.8\columnwidth]{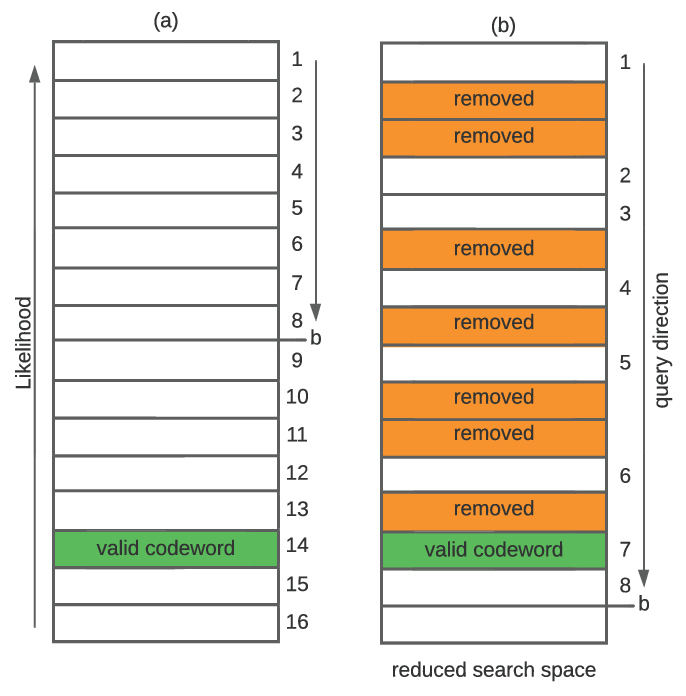}
     \caption{A sketch showing a stack of candidate sequences sorted in descending order with respect to a likelihood metric (the codeword at the top has the highest likelihood). With no abandonment condition, removing the invalid sequences accelerates reaching the first valid codeword by fewer queries (7 queries versus 14 queries). With abandonment after $b=8$ queries, case (a) will fail to reach the valid codeword.}
    \label{fig:success_query}
    \vspace{-10pt}
\end{figure}


\section{Implementation Considerations}\label{sec:implement}
In this section, we propose a hardware-compatible procedure illustrated in Algorithm \ref{alg:dist_parts} to efficiently perform the first and second levels of weight partitioning with the required number of parts. An example of integer partitioning of $w=18$ into $t=4$ distinct parts is illustrated in Fig. \ref{fig:ex_k_parts_distinct}. We use this example along with Algorithm \ref{alg:dist_parts} to explain the procedure. 
\begin{figure}
    \centering
    \includegraphics[width=0.7\columnwidth]{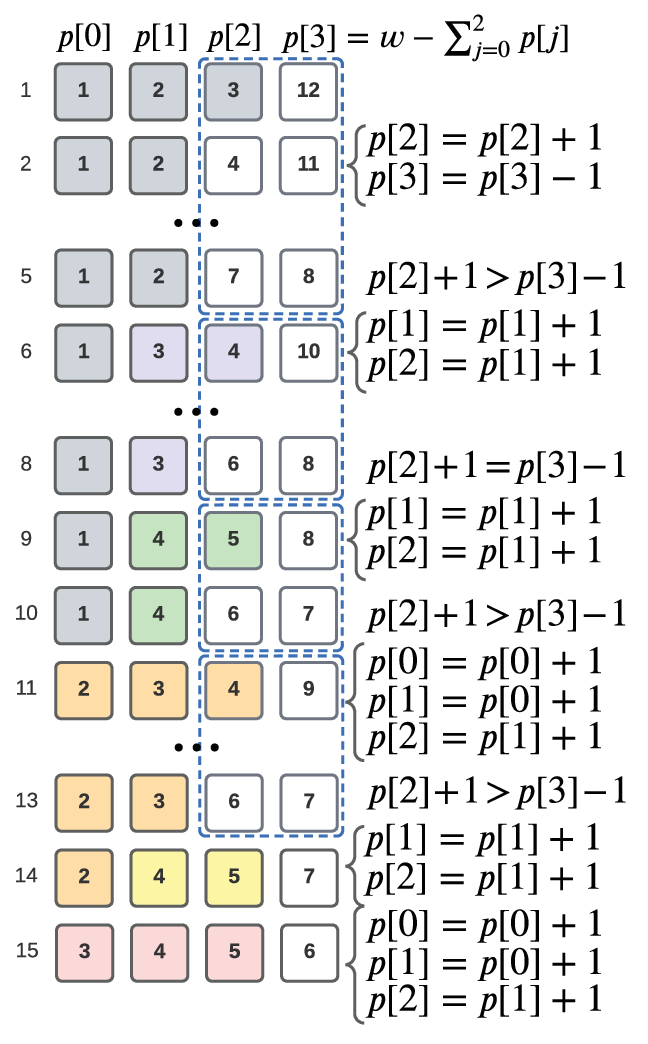}
    \caption{An example showing integer partitioning procedure for $w=18$ into four distinct integers, $k=4$.}
    \label{fig:ex_k_parts_distinct}
\end{figure}
The procedure for every integer $w=w_L^{(j)},j\in[1,p]$ starts with an initial sequence $\bp$ of $t$ elements as performed in lines 2-3 of Algorithm \ref{alg:dist_parts}. 
Before the generation of the next sequence of integer parts, we check to see which of the following two operations should be sought.  
\begin{enumerate}
    \item Increment-and-decrement: If we have $p[t-2]+1 < p[t-1]-1$, we keep the sub-sequence $p[0],p[1],\dots,p[t-3]$ while incrementing $p[t-2]=p[t-2]+1$ and decrementing $p[t-1]=p[t-1]-1$. These operations are performed in the last two parts in white cells, circled by blue dashed lines in Fig. \ref{fig:ex_k_parts_distinct} except for the first sequence in the circle that plays the role of the basis for these operations. As long as $i=1$ in for loop in lines 12-23 in Algorithm \ref{alg:dist_parts}, this operation continues to generate new sequences. Resuming this loop is performed by line 21. 
    Note that the assignment in line 14 of Algorithm \ref{alg:dist_parts} is the general form for any $i$. For instance, we can get line 12 by substituting $i=1$ in line 14. Here, we showed them separately because we predominantly have $i=1$. 
        \begin{figure}[h]
            \centering
            \includegraphics[width=0.7\columnwidth]{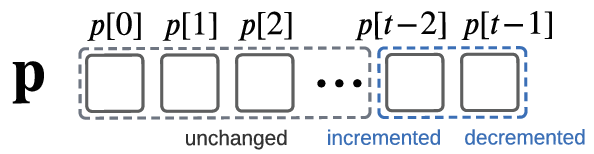}
            \label{fig:incr_decr}
        \end{figure}
        
    \item Re-initialization: If we have $p[t-2]+1 \geq p[t-1]-1$, we would have non-distinct parts in the sequence in the case of equality or repeated sequence when inequality holds. Hence, we need to change the other parts, i.e., $p[t-1-i],t-1\leq i\leq 2$. The extent of change is determined by some $i>1$ such that the condition in line 15 is met. 
    
        \begin{figure}[h]
            \centering
            \includegraphics[width=0.9\columnwidth]{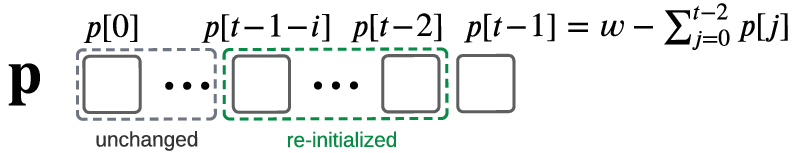}
            \label{fig:re_init}
        \end{figure}
        The re-initialization for such an $i$ will be as follows:
        $$[p[t\!-\!1\!-\!i]\!+\!1,p[t\!-\!1\!-\!i]\!+\!2,\dots,p[t\!-\!1\!-\!i]\!+\!i\!+\!1]$$
        For instance, in Fig. \ref{fig:ex_k_parts_distinct}, the sequences 6,9, and 14 are re-initialized when $i=2$ and the sequences 11 and 15 when $i=3$. Note that when $i=t-1$, i.e., all parts except for $p[t-1]$ are re-initialized and still the condition $p[t-2]+1 \geq p[t-1]-1$ in line 15 is not met, the process ends. This means all the possible options for parts have been checked.
\end{enumerate} 
As mentioned in Section \ref{sec:prelim}, no parts can be larger than the length of the code. Here, we need to consider this as well for the length of the segment denoted by $p_{\max}$ in Algorithm \ref{alg:dist_parts} as you can observe in lines 6 and 18. 

A similar procedure can be used for the first level of integer partitioning for the error pattern bases by lifting the constraint on the distinctness of the parts and allowing the permutation. However, we need to consider the minimum sub-weight (1 or 3 depending on $s_j$) that each segment can take. Given these differences, one can observe that the initialization of non-frozen segments can allow repetition of 1's or 3's, instead of distinct values of $1,2,3,\cdots$. For instance, for three segments with $s_1=s_2=1,s_3=0$ and base $[1\;1\;1]$, we can start the above procedure for $w_L=7$ with $\bp=[1\;1\;5]$, then we proceed with $\bp=[1\;2\;4]$, $\bp=[1\;3\;3]$. The rest are $\bp=[2\;1\;4]$, $\bp=[2\;2\;3]$, and finally, $\bp=[3\;1\;3]$. In this example, the order of re-initialization is the same as Fig. \ref{fig:ex_k_parts_distinct}.

\begin{algorithm}[t]
    \footnotesize
    \caption{Non-recursive integer partitioning to a fixed number of distinct parts}
    \label{alg:dist_parts} 
    \DontPrintSemicolon
    \SetKwInOut{Input}{input}
    \SetKwInOut{Output}{output}
    \Input{sub-weight $w$, part size $t$, largest part $p_{max}=n$ }  
    \Output{$\cP$}
    
    $\cP \gets \{\}$\;
    $\bp \gets [1, 2, \dots, t-1]$\; 
    $\bp \gets \bp+[w-\text{sum}(\bp)]$\;
    \If{$p[t] <= p[t-1]$}{
        \KwRet $\cP$\;
    }
    \If{$p[t] \leq p_{\max}$}{
        $\cP \gets \cP \cup \{\bp\}$\;
    }
    \text{incr\_decr } $\gets $ True \tcp*{Operation: Interment-and-decrement }
    \While{\emph{True}}{
        \For{$i$ \emph{in} $[1,t-1]$}{
            \eIf{$i=1$}{
                $p^* \gets p[t-1]-1$
            }{
                $p^* \gets w\!-\!(i\!\cdot\!p[t\!-\!1\!-\!i]\!+\!\sum_{\!j=\!1}^i j) \!-\! \sum_{j\!=\!0}^{t\!-\!2\!-\!i}p[j]$\;
            }
            \eIf{$p[t-1-i]+i < p^*$}{
                $\bp \gets \bp[0\!:\!t\!-\!2\!-\!i]\!+[ p[t-1-i]\!+\!1,p[t\!-\!1\!-\!i]+2,\dots,p[t-1-i]+i+1]$\;
                $\bp \gets \bp + [w-\text{sum}(\bp)]$\;
                
                \If{$p[t-1] \leq p_{\max}$}{
                    $\cP \gets \cP \cup \{\bp\}$\;
                }
                \text{incr\_decr} $\gets $ True\;
                break\;
            }{
                \text{incr\_decr} $\gets $ False\;
            }
        }
        \If{$\emph{incr\_decr} =$ \emph{False} $\emph{ \bf and } i = t-1$}{
            break\;
        }
    }
    \KwRet $\cP$\;
\end{algorithm}

\section{Numerical Results and Discussion}\label{sec:num_results}
We consider two sample codes for the numerical evaluation of the proposed approach. The polarization-adjusted convolutional (PAC) code (64,44) \cite{arikan2} is constructed with Reed-Muller-polar rate-profile with design-SNR=2 dB and convolutional generator polynomial $[1,0,1,1,0,1,1]$. The extended BCH code (128,106) with the primitive polynomial $D^7+D^3+1$ and $t=3$. 
Note that the rows $\bh_{1}$ and $\bh_{2}$ in $\bH$ matrix for eBCH code (128, 106) satisfy the relationship $\supp(\bh_{2})\subset\supp(\bh_{1})$ where $\bh_2=\bone$ and $|\bh_1|/2=|\bh_2|=64$. Hence, for two constraints, we modify $\bh_1$ by $\bh_1=\bh_1\oplus\bh_2$. Similarly, the rows $\bh_{1}$, $\bh_{4}$, and $\bh_{5}$ in $\bH$ matrix for PAC code (64, 44) satisfy the relationship $\supp(\bh_{5})\subset\supp(\bh_{4})\subset\supp(\bh_{1})$ where $\bh_1=\bone$ and $|\bh_1|/2=|\bh_4|=2|\bh_5|=32$. The Python implementation of the proposed algorithm can be found in \cite{code}.

\subsection{Performance vs Queries}
Figs. \ref{fig:bler-pac} and \ref{fig:bler-ebch} show the block error rates (BLER) of the PAC  code (64, 44) and the extended BCH code (128,106), respectively, under the ORBGRAND with no constraints (NoC) and the segmented GRAND with the maximum number of queries based on \eqref{eq:check}, a.k.a abandonment threshold, $b=10^5,10^6$ . Note that the threshold $b$ in GRAND algorithms should be approximately $2^{n-k}$ queries \cite[ Theorem 2]{duffy-tit} to find the error pattern and get reasonable performance. 

\begin{figure}
    \centering
    \includegraphics[width=0.9\columnwidth]{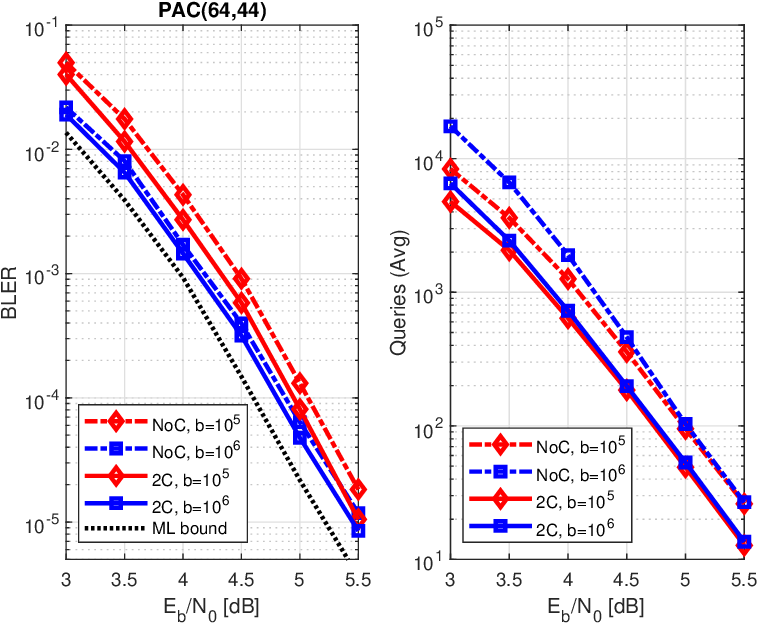}
    \caption{Performance comparison between three sub-pattern generators based on three constraints (3C) and a single generator with no constraints (NoC). 
    The vertical axis is on the logarithmic scale for both queries and BLER.}
    \label{fig:bler-pac}
\end{figure}

\begin{figure}
    \centering
    \includegraphics[width=0.9\columnwidth]{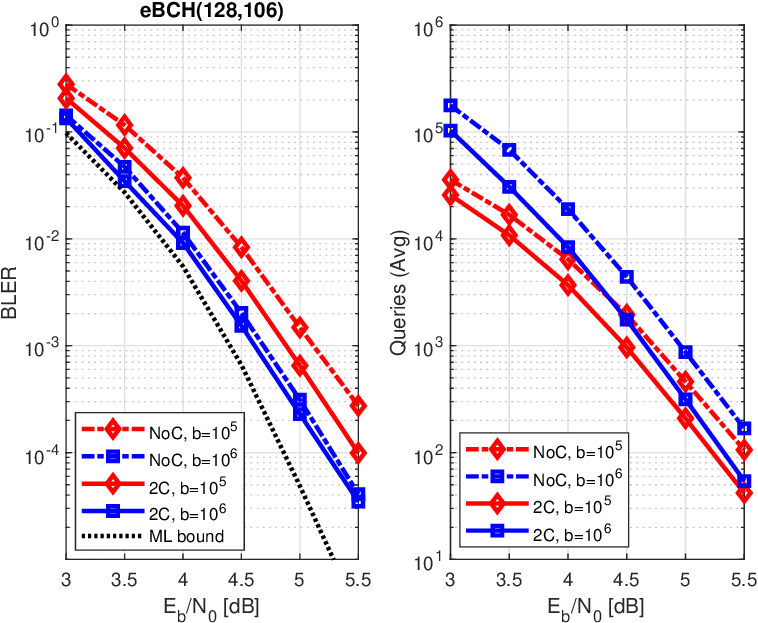}
    \caption{Performance comparison between three sub-pattern generators based on two constraints (2C) and a single generator with no constraints (NoC).} 
    \label{fig:bler-ebch}
\end{figure}

As expected, the average queries reduce significantly for both codes under segmented ORBGRAND. In the case of the PAC code (64,44), the average queries become half at high SNR regimes, while this reduction is larger at low SNR regimes. The reduction in average queries for eBCH(128,106) is more significant under the same abandonment thresholds as the short PAC code. Note that average queries for the short PAC code with different $b$'s are approaching at high SNR regimes due to the effectiveness of smaller $b$ at this code length. Furthermore, there is a BLER improvement where $b=10^5$; however, this improvement diminishes by increasing $b$ or under no abandonment as we will observe later. Note that unlike the comparisons in \cite{rowshan-const_GRAND} where the BLER was fixed and the impact of applying constraints on the average queries was studied, here we fix the maximum number of queries $b$ for both ORBGRAND and segmented ORBGRAND to have a fair comparison. As discussed in Section \ref{sec:complx}, in case of decoding failure by ORBGRAND, if we reduce the search space, we don't have to process many invalid error patterns. As a result, the first valid pattern may fall within the abandonment threshold $b$, and the segmented ORBGRAND would succeed. 
In the table below, we show the average queries of two codes at $E_b/N_0=5$ dB (with two/three constraints, denoted by 2C/3C, and with no constraints/segmentation, denoted by NoC) for the maximum queries of $b=10^4,10^5$. The average queries are reduced by halves (in the case of two segments, it is slightly less than half, while in the case of three segments, it is more than half).
\begin{center}
\footnotesize
\begin{tabular}{c|cc|cc} 
 \hline
 & \multicolumn{2}{c}{PAC(64,44)} & \multicolumn{2}{c}{eBCH(128,106)} \\ 
 \hline
 & NoC & 3C & NoC & 2C  \\ 
 \hline
$b=10^5$ & 95.1 & 49.0 & 460.7 & 208.9  \\ 
 \hline
$b=10^6$ & 103.3 & 	53.2 &	872.7 & 314.9
  \\ 
 \hline
\end{tabular}
\end{center}
Note that if we maintain the BLER, the average query reduction is expected to approximately follow Lemma \ref{lma:seach_space_size} as it was shown numerically in \cite{rowshan-const_GRAND}.
Note that here, with abandonment threshold, further reduction to meet the expectation in Lemma \ref{lma:seach_space_size} is traded with BLER improvement.


Now, let us consider ORGBGRAND without abandonment. Fig. \ref{fig:bler-q-ebch-compar} compares the BLER and the (average) complexity of eBCH(128,106) under various decoding algorithms. The main benchmark is naturally ORBGRAND. Compared to ORBGRAND, segmented ORBGRAND reduces the average number of queries by three times, while BLER remains almost the same as before. 

We also compare it with the most popular MRB-based decoding algorithm, that is, ordered statistics decoding (OSD) with order $i$, as its relationship with its variants such as the box-and-match algorithm (BMA) \cite{valembois} and enhanced BMA \cite{jin} is known. Moreover, the reduction in the complexity of the variant comes at the cost of the increase in space complexity which makes the comparison unfair. For instance, the BMA reduces the computational complexity of OSD roughly by its squared root at the expense of memory, as the BMA with order $i$ considers all error patterns of weight at most $2i$ over $s$ most reliable positions ($s>k$). The BLER of OSD(2) is remarkable compared to other algorithms while it provides a reasonable complexity at low SNR regimes. Whereas ORBGRAND requires considerably fewer queries at high SNR regimes at the cost of degradation in BLER performance.

The other two algorithms used for comparison are Berlekamp-Massey Algorithm and Chase-II algorithm. Chase-II algorithm, denoted by Chase-II($t$), for decoding a code with the error-correcting capability of $t$ has the computational complexity of order $2^t\cdot O$(HD) as it uses a hard decision (HD) decoder, such as the Berlekamp-Massey Algorithm with the complexity of order $O(n^2)$,  in $2^t$ times as the decoder attempts all the error patterns with weight up to $t=\lfloor \frac{d_{min}-1}{2} \rfloor$ over the $t$ least reliable positions, hence, $\sum_{j=0}^{t} {t \choose j}=2^t$. In the case of eBCH(128,106), we have $t=3=\lfloor \frac{d_{min}-1}{2}  \rfloor$ where $d_{min}=7$. As can be seen, the BLER of the Berlekamp-Massey Algorithm and Chase-II algorithm is not comparable with OSD and ORBGRAND though they have a computational complexity of orders $O(2^{14})$ and $8\cdot O(2^{14})$, respectively. Furthermore, 
we observed that by increasing the total attempts to $2^t=2^8$, the Chase-II algorithm can approach the BLER of ORBGRAND as sown in Fig. \ref{fig:bler-q-ebch-compar}.

Furthermore, we use the early termination criterion as discussed in Section \ref{ssec:OSD}. This remarkably reduces the average queries of OSD(2) as shown in Fig. \ref{fig:bler-q-ebch-compar}. Lastly, we find the ML bound as follows: The ``ML bound" is determined by identifying instances where the optimal ML decoder would fail. During the simulations, each time a decoding error occurred, we compared the likelihood of the decoded codeword with that of the transmitted codeword. Specifically, we checked if the likelihood of the received signal given the decoded codeword, $W(\br \mid \textup{x}(\hat{\bc}))$, exceeded the likelihood of the received signal given the actual transmitted codeword, $W(\br \mid \textup{x}(\bc))$. If $W(\br \mid \textup{x}(\hat{\bc}))>$ $W(\br \mid \textup{x}(\bc))$, the ML decoder would also misinterpret the received signal and produce the same decoding error. Here, we use the squared Euclidean distance \eqref{eq:sq_euclid_dist} as a measure of likelihood.  This process allows us to estimate the performance bound of an ML decoder by identifying cases where any decoder, including the optimal one, would fail. As can be seen, the gap between the ML bound and OSD(2) in the high SNR regime is negligible. According to our observation, OSD(3) performance almost overlaps with the ML bound.

\begin{figure}
    \centering
    \includegraphics[width=0.9\columnwidth]{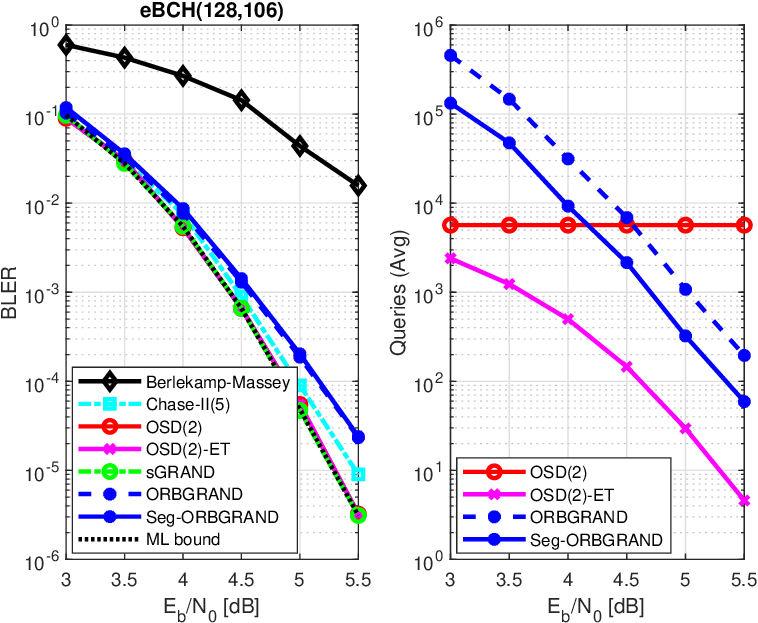}
    \caption{Performance and complexity comparisons of various decoding algorithms for eBCH(128,106).}
    \label{fig:bler-q-ebch-compar}
\end{figure}

\subsection{Complexity: 
Number of Operations}
The average number of queries per decoding may not fully capture the complexity of finding the error pattern for two reasons: (1) the computational complexity of each query in ORBGRAND differs from that in OSD($i$), and (2) the initial computational overhead, such as elementary row operations in OSD, is not the same in the two algorithms. The latter becomes especially significant at high SNR, where the average number of queries drops and the initial computations may dominate. In this section, we consider the average number of operations required to decode each received sequence. While these measures are implementation-dependent, they provide insight into the efficiency and true reduction in complexity achieved by the proposed scheme. The basic operations considered are addition, comparison, multiplication, and exclusive-OR (XOR) operation as the basic bit-wise operation. Note that we do not consider variable or vector assignments (value loading). 
To facilitate the comparison of algorithms, we convert the number of non-addition operations to addition-equivalent operations. Although arithmetic operations can be implemented differently in hardware, for simplicity, we consider ripple-carrying addition, the comparison implemented by bit-wise comparison of the most significant bit to least significant, or alternatively by subtraction and sign checking, and naive binary multiplication implemented by shifting the multiplier (one bit at a time, in total $m$ times) and adding the shifted multiplicand based on the multiplier's bits, similar to long multiplication. Table~\ref{tab:basic_operations} lists the order of complexity of these operations and the multiplicative factor used to obtain the addition-equivalent of the corresponding operation. For example, every multiplication of two $m$-bit numbers is equivalent to $m$ addition operations. 
Furthermore, we assume the use of a sort algorithm and initial transformation of the generator matrix with the complexities of order $O(n\log_2 n)$ comparisons and $O(n \cdot \min(k, n-k)^2)$ \cite{fossorier} bit-wise operations, respectively, where $n$ is the sequence length and $k$ is the code dimension. 
Note that segmentation by two would halve the length $n$, as discussed in Section \ref{sec:err_subpattern}. Hence, to find the complexity of decoding in terms of operations, we count all these operations individually in every decoding attempt and then convert the non-addition operations to addition-equivalent operations based on the multiplicative factors in Table \ref{tab:basic_operations}.

\begin{table}
\centering
\footnotesize
\caption{Complexity of basic operations expressed as multiples of addition ($m$ is the number of bits).} \label{tab:basic_operations}
\begin{tabular}{c|c|c}
 & Complexity & \makecell{Multiplicative \\ Factor to Addition} \\
 \hline
Addition & $O(m)$ & 1 \\
Comparison & $O(m)$ & 1 \\
Multiplication & $O(m^2)$ & $m$ \\
eXclusive-OR & $O(1)$ & $1/m$ \\
\end{tabular}
\end{table}

Fig. \ref{fig:op-ebch-compar} illustrates the average total of all operations expressed in terms of equivalent addition operations for each decoding attempt, assuming that the real numbers are represented by $m=6$ bits. This is effectively a translation of the average queries in Fig. \ref{fig:bler-q-ebch-compar} into the average number of addition-equivalent operations. As shown, segmented ORBGRAND exhibits the lowest complexity in terms of average operations for $E_b/N_0\geq3$. Notably, at BLER$>10^{-2}$ (although it is not a desirable level) where the power gain/difference is small in Fig. \ref{fig:bler-q-ebch-compar}, the difference between the average number of operations in segmented ORBGRAND and OSD(2)-ET is not significant. 
For OSD with early termination, the primary contributor to the complexity of every query is the computation of the likelihood metric, such as the squared Euclidean distance and the correlation discrepancy, which is computationally cheaper and is used here, for each candidate codeword, whereas each query (generation of a new error pattern) in ORBGRAND is performed through a few addition-equivalent operations on the current logistic weight, as illustrated in Fig. \ref{fig:ex_k_parts_distinct} or Algorithm \ref{alg:dist_parts}. This process, before generating a new error pattern, may be equivalent to or slightly more complex than the process of keeping track of generating distinct error patterns in OSD. Note that the logistic weight in ORBGRAND serves to guide and track the generation of test error patterns in a specific order. This is different from the likelihood metric in OSD, which is used after error pattern generation for comparison purposes. 

Although we can use some properties to reduce the complexity of computing the likelihood metric by taking advantage of difference between the error patterns, but that comes with computational overhead which contributes to latency. 
Observe that as the average number of queries approaches the lower bound of 1 for order-0 decoding, computing the likelihood metric and the initial computational overhead, including Gaussian elimination, dominate the overall complexity. Consequently, we see the slope of the pink curve flattening in the high SNR regime.

\begin{figure}
    \centering
    \includegraphics[width=0.8\columnwidth]{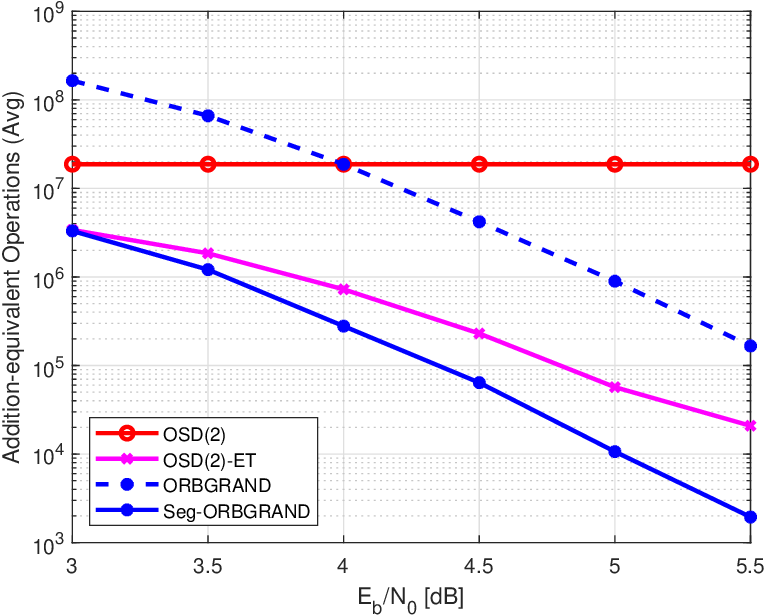}
    \caption{The average number of addition-equivalent operations for various decoding algorithms for eBCH(128,106).}
    \label{fig:op-ebch-compar}
\end{figure}

The reduction in the average number of operations in segmented ORBGRAND is due to the reuse of generated sub-patterns in multiple error patterns. 
This saves a significant number of operations. However, this approach requires a pool of pregenerated sub-patterns. The maximum pool size used for the above results is $248$ sub-patterns. It can also be implemented without this pool by combining every generated sub-pattern of one segment with all the sub-patterns of other segment(s) that meet the sub-weight and overall weight. The design of an efficient hardware architecture can be the subject of future work. 

As mentioned in Section \ref{sec:complx}, the number of queries required to achieve ML performance with GRAND is of the order of $2^{n-k}$ \cite[Theorem 2]{duffy-tit}. This makes GRAND more suitable for very high code rates. Consequently, as the code rate decreases, a significant increase in complexity is expected. If we limit the number of possible queries by imposing an abandonment threshold, the performance would significantly degrade. Let us examine this by considering eBCH(128,99,10), which has a slightly lower code rate compared to eBCH(128,106). Here, we limit the queries to $b=10^7$. As shown in Fig. \ref{fig:op-ebch-compar_128_99}, both the gap between the BLER curves and the average number of operations increase. This aligns with our expectation and indicates that GRAND is recommended only for very high-rate codes, where a good performance can be obtained with relatively low average complexity and small  abandonment threshold. 



 


In terms of relative decoding speedup (the difference in the decoding time) from high SNR to low SNR points ($E_s/N_0=0.5-2.5$ dB almost equivalent to $E_b/N_0=3.5-5.5$ for $R=0.5$), we compare our work with \cite{choi2020fast} where about 100 times decoding speedup for the (127,43) code was reported, the speedup of our work is also about the same, as the table above shows. Note that a comparison of the decoding time in a fair way is not possible as the reported time depends on the CPU clock frequency, cache size and architecture, system load and configuration, the choice of programming language and its associated compiler, etc.


\begin{figure}
    \centering
    \includegraphics[width=0.9\columnwidth]{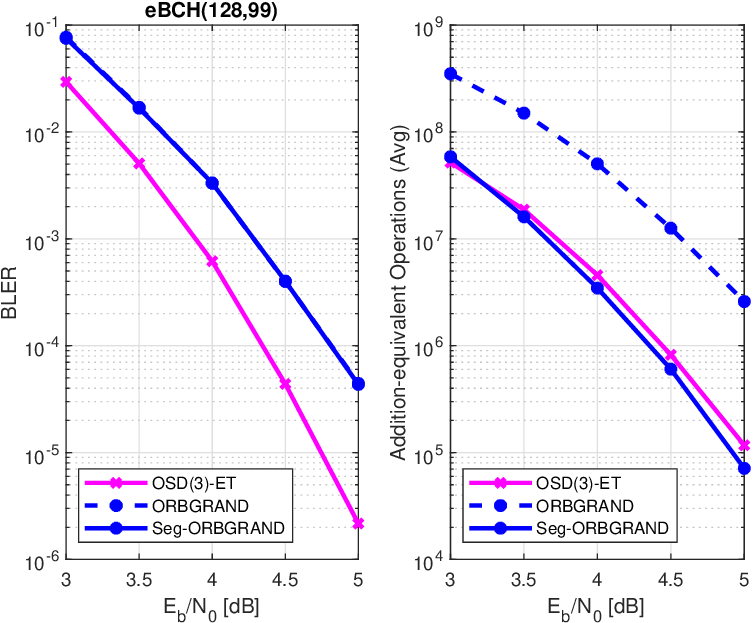}
    \caption{Performance and complexity comparisons of various decoding algorithms for eBCH(128,99).}
    \label{fig:op-ebch-compar_128_99}
\end{figure}

\section{CONCLUSION} 
In this paper, we propose an approach to divide the search space for the error sequence induced by channel noise through segmentation. Each segment is defined based on parity constraints extracted from the parity check matrix. We then employ multiple error pattern generators, each dedicated to one segment. We introduce a method to combine these sub-patterns in a near-ML order for checking. Since this approach generates valid error patterns with respect to the selected parity constraints, both the average number of queries and the block error rate (BLER) performance (under abandonment only) improve significantly. Additionally, the reuse of pre-generated sub-patterns in forming new error patterns reduces the number of operations for each query. Consequently, alongside the reduction in the average number of queries, the decoding time decreases considerably, down to one-fifth of ORBGRAND. The study of the tradeoff between memory requirements based for various scheduling schemes on one hand and throughout on the other remians as future work for hardware architecture design.








\end{document}